\documentclass[runningheads]{llncs}

\usepackage[subtle]{savetrees}

\usepackage[utf8]{inputenc}

\usepackage[]{auxlib}
\renewcommand*{\translate}[1]{#1}    % I get local errors without this, so disable automatic translations for now (don't use it here anyway)

\usepackage{todonotes}
\usepackage{xcolor}

\usepackage[font={small}]{caption}
\usepackage{subcaption}

\usepackage{footmisc}
\usepackage{dsfont}

\usepackage{hyperref}
\usepackage[doipre={doi:~}]{uri}
\usepackage[capitalize,nameinlink,noabbrev]{cleveref}

\usepackage[numbers, sort&compress,]{natbib}

\usepackage{enumitem}
\usepackage{multicol}

\usepackage{thmtools}
\usepackage{thm-restate}

\usepackage{bibretriever}

\definecolor{j1_color}{RGB}{128, 128, 255}
\definecolor{j2_color}{RGB}{128, 177, 128}
\definecolor{j3_color}{RGB}{197, 128, 128}

\newcommand{\Alpha}{A}
\newcommand{\Beta}{B}
\newcommand{\Rho}{P}
\newcommand{\dt}[0]{\ensuremath{\dif{t}}}
\newcommand{\dalpha}[0]{\ensuremath{\dif{\alpha}}}

\newcommand*{\OPTM}{M^*}

\makeatletter
\newcommand\footnoteref[1]{\protected@xdef\@thefnmark{\ref{#1}}\@footnotemark}
\makeatother

\def\WFalg/{\textsc{WaterFill}}
\def\WFstep/{\textsc{WFstep}}

\newcommand*{\wl}{\operatorname{wl}}

\createcustomnote{christoph}{backgroundcolor=AUXLyellow, prefix=Christoph}
\createcustomnote{peter}{backgroundcolor=AUXLlightgreen, prefix=Peter}
\createcustomnote{florian}{backgroundcolor=AUXLlightgray, prefix=Florian}

\renewcommand{\enquote}[1]{\glqq{#1}\grqq{}}

\begin{document}

\title{Improved Scheduling with a Shared Resource}
 \author{
     Christoph Damerius\inst{1}$^*$,
     Peter Kling\inst{1},
     Florian Schneider\inst{1}
 }
 \authorrunning{C.~Damerius, P.~Kling and F.~Schneider}
 \institute{
     $^1$University of Hamburg, 22527 Hamburg, Germany\\
     \email{peter.kling@uni-hamburg.de}, \email{fschneider@informatik.uni-hamburg.de}\\
     $^*$ Corresponding author: \email{damerius@informatik.uni-hamburg.de}\\
 }

\maketitle

\begin{abstract}
We consider the following shared-resource scheduling problem:
Given a set of jobs $J$, for each $j\in J$ we must schedule a job-specific processing volume of $v_j>0$.
A total resource of $1$ is available at any time.
Jobs have a resource requirement $r_j\in\intcc{0,1}$, and the resources assigned to them may vary over time.
However, assigning them less will cause a proportional slowdown.

\hspace{1.7em}We consider two settings.
In the first, we seek to minimize the makespan in an online setting:
The resource assignment of a job must be fixed before the next job arrives.
Here we give an optimal $e/(e-1)$-competitive algorithm with runtime $\LDAUOmicron{n\text{\,\,} \log n}$.
In the second, we aim to minimize the total completion time.
We use a continuous linear programming (CLP) formulation for the fractional total completion time and combine it with a previously known dominance property from malleable job scheduling to obtain a lower bound on the total completion time.
We extract structural properties by considering a geometrical representation of a CLP's primal-dual pair.
We combine the CLP schedule with a greedy schedule to obtain a
$(3/2+\varepsilon)$-approximation for this setting.
This improves upon the so far best-known approximation factor of $2$.
\end{abstract}


\section{Introduction}%
\label{sec:introduction}

Efficient allocation of scarce resources is a versatile task lying at the core of many optimization problems.
One of the most well-studied resource allocation problems is parallel processor scheduling, where a number of \emph{jobs} need (typically at least temporarily exclusive) access to one or multiple \emph{machines} to be completed.
%Here, jobs and machines can be actual computational jobs and computers but may also represent more abstract processes (e.g., tasks or experts in a company).
The problem variety is huge and might depend on additional constraints, parameters, available knowledge, or the optimization objective (see~\cite{DBLP:reference/crc/2004sch}).

In the context of computing systems, recent years demonstrated a bottleneck shift from \emph{processing power} (number of machines) towards \emph{data throughput}.
Indeed, thanks to cloud services like AWS and Azure, machine power is available in abundance while data-intensive tasks (e.g., training LLMs like ChatGPT) rely on a high data throughput.
If the bandwidth of such data-intensive tasks is, say halved, they may experience a serious performance drop, while computation-heavy tasks care less about their assigned bandwidth.
In contrast to the number of machines, throughput is (effectively) a \emph{continuously} divisible resource whose distribution may be easily changed \emph{at runtime}.
This opens an opportunity for adaptive redistribution of the available resource as jobs come and go.
Other examples of similarly flexible resources include power supply or the heat flow in cooling systems.

This work adapts formal models from a recent line of work on such flexible resources~\cite{DBLP:conf/cocoa/DameriusKLSZ20, DBLP:journals/scheduling/AlthausBKHNRSS18, DBLP:conf/spaa/KlingMRS17} and considers them under new objectives and settings.
Classical \emph{resource constrained scheduling}~\cite{DBLP:journals/algorithmica/NiemeierW15, DBLP:journals/siamcomp/GareyJ75, DBLP:conf/soda/JansenMR16, DBLP:conf/esa/MaackPR22} assumes an \enquote{all-or-nothing} mentality (a job can be processed if it receives its required resource but is not further affected).
One key aspect of the model we consider is the impact of the amount of received resource on the jobs' performance (sometimes referred to as \emph{resource-dependent processing times}~\cite{DBLP:journals/orl/Kellerer08, DBLP:conf/soda/JansenMR16, DBLP:journals/mp/GrigorievSU07, DBLP:journals/disopt/GrigorievU09}).
The second central aspect is that we allow a job's resource assignment to change while the job is running.

\vspace{-0.04em}
\subsection{Model Description and Preliminaries}%
\label{sec:preliminaries}

We consider a scheduling setting where a set $J = \intcc{n} \coloneqq \set{1, 2, \dots, n}$ of $n \in \N$ \emph{jobs} compete for a finite, shared resource in order to be processed.
A \emph{schedule} $R = (R_j)_{j \in J}$ consists of an (integrable) function $R_j\colon \R_{\geq0} \to \intcc{0, 1}$ for each $j \in J$ (the job's \emph{resource assignment}) that returns what fraction of the resource is assigned to $j$ at time $t \in \R_{\geq0}$.
%Given such a schedule $R$ for $J$,
We use $R(t) = ( R_j(t) )_{j \in J}$ to refer to $j$'s \emph{resource distribution} at time $t$ and $\bar{R}(t) \coloneqq \sum_{j \in J} R_j(t)$ for the \emph{total resource usage} at time $t$.
Each $j \in J$ comes with a (\emph{processing}) \emph{volume} $v_j \in \R_{\geq0}$ (the total amount of resource the job needs to receive over time in order to be completed) and a \emph{resource requirement} $r_j \in \intcc{0, 1}$ (the maximum fraction of the resource the job can be assigned).
We say a schedule $R = (R_j)_{j \in J}$ is \emph{feasible} if:
\vspace{-0.12em}
\begin{itemize}[noitemsep]
\item the resource is never overused: $\forall t \in \R_{\geq0}\colon \bar{R}(t) \leq 1$,
\item a job never receives more than its resource requirement: $\forall t \in \R_{\geq0}\colon R_j(t) \leq r_j$, and
\item all jobs are completed: $\forall j \in J\colon \int_{0}^{\infty} R_j(t) \dif{t} \geq v_j$.
\end{itemize}
\vspace{-0.12em}
For $j \in J$ we define its \emph{processing time} $p_j \coloneqq v_j / r_j$ as the minimum time that $j$ requires to be completed.
See \cref{fig:scheduling-example} for an illustration of these notions.

For a schedule $R = (R_j)_{j \in J}$ we define $C_j(R) \coloneqq \sup\set{t \geq 0 | R_j(t) > 0}$ as the \emph{completion time} of job $j \in J$.
We measure the quality of a schedule $R$ via its \emph{makespan} $M(R) \coloneqq \max\set{C_j(R) | j \in J}$ and its \emph{total completion time} $C(R) \coloneqq \sum_{j \in J} C_j(R)$.
Our analysis additionally considers the \emph{total fractional completion time} $C^F(R) \coloneqq \sum_{j \in J} C^F_j(R)$, where $C^F_j(R) \coloneqq \int_{0}^{\infty} R_j(t) \cdot t / v_j \dt$ is job $j$'s \emph{fractional completion time}.
%(the average completion time of its workload).

\begin{figure}[t]
\begin{subfigure}{0.48\linewidth}
\centering\includegraphics[height=8em]{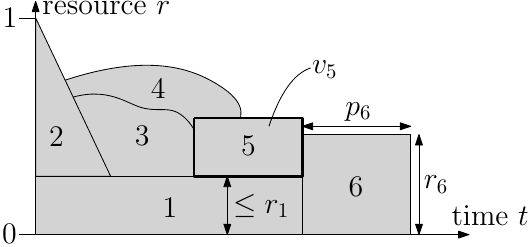}
\caption{
    A schedule for six jobs.
    The resource assignment of job $j$ is given by the height of $j$'s area at time $t$ and must never exceed $r_j$.
    The total area of a job is equal to its volume.
}%
\label{fig:scheduling-example}
\end{subfigure}
\hfill
\begin{subfigure}{0.48\linewidth}
\centering\includegraphics[height=8em]{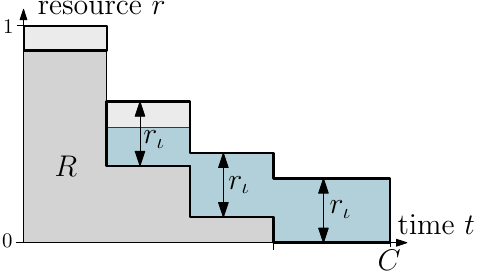}
\caption{
    Augmenting a schedule $R$ by a job $\iota$ via $\WFstep/(R, \iota, C)$.
    The volume of $\iota$ is \enquote{poured} into the fat-outlined area.
    The blue area indicates where it is eventually scheduled.
}\label{fig:WaterFill}
\end{subfigure}
\caption{}
\vspace{-1.3em}
\end{figure}

\paragraph{Relation to Malleable Tasks with Linear Speedup}
Our problem assumes an arbitrarily divisible resource, as for example the bandwidth shared by jobs running on the same host.
Another common case are jobs that compete for a \emph{discrete} set of resources, like a number of available processing units.
This is typically modeled by a scheduling problem where a set $J$ of $n$ \emph{malleable} jobs of different sizes $s_j$ (length when run on a single machine) must be scheduled on $m$ machines.
Each machine can process at most one job per time, but jobs $j$ can be processed on up to $\delta_j \in \intcc{m}$ machines in parallel with a linear speedup.
Jobs are preemptable, i.e., they can be paused and continued later on, possibly on a different number of machines.
See~\cite[Ch.~25]{DBLP:reference/crc/2004sch} for a more detailed problem description.

This formulation readily maps to our problem by setting $j$'s processing volume to $v_j = s_j / m$ and its resource requirement to $r_j = \delta_j / m \in \intoc{0, 1}$.
The only difference is that our schedules allow for arbitrary resource assignments, while malleable job scheduling requires that each job $j$ gets an \emph{integral} number $\delta_j$ of machines (i.e., resource assignments must be multiples of $1/m$).
However, as observed by \textcite{beaumont2012minimizing}, fractional schedules can be easily transformed to adhere to this constraint:
\begin{observation}[{\cite[Theorem~3, reformulated]{beaumont2012minimizing}}]%
\label{obs:cannonical_schedules}
Consider a feasible schedule $R$ for a job set $J$ in which $j \in J$ completes at $C_j$.
Let $m \coloneqq 1 / \min\set{r_j | j \in J}$.
We can transform each $R_j$ without changing $C_j$ to get
\begin{math}
R_j(t)
\in
\set{i/m | i \in [m] \cup \set{0}}
\end{math}
for any $t \in \R_{\geq0}$ and such that each $R_j$ changes at most once between consecutive completion times.
\end{observation}

We first consider \emph{online makespan minimization} (\cref{sec:minmakespan}), where the scheduler must commit to future resource assignments as jobs arrive (as in list-scheduling).
Afterwards, we consider \emph{offline total completion time minimization} (\cref{sec:mintotalcompletiontime}).

\subsection{Related Work}%
\label{sec:relatedwork}

% \footnote{%
%     While we use the most common interpretations, note that some publications use these terms slightly differently, e.g., use the term \enquote{malleable} for what we call \enquote{moldable}.
% }

%In the following we provide a brief overview of several related resource-constrained scheduling problems.
% We focus on problems that share the same key concepts of \emph{continuously shareable resources} as our own problem as well as on the closely related, more classical job scheduling problems for \emph{rigid, moldable, and malleable} jobs.

Our model falls into the class of continuous shared-resource job scheduling as introduced in \cite{DBLP:journals/scheduling/AlthausBKHNRSS18} and its variants~\cite{DBLP:conf/cocoa/DameriusKLSZ20, DBLP:conf/spaa/KlingMRS17}.
These models have the same relation between a job's resource requirement, the assigned resource, and the resulting processing time as we but
only consider makespan minimization as objective.
The two main differences are that they assumed an additional constraint on the number of machines and considered discrete time slots in which resource assignments may not change.

Another closely related model is \emph{malleable} job scheduling, where the
number of machines assigned to a job can be dynamically adjusted over time.
%But now the scheduler may adjust the number of machines assigned even when the job is running.
If each job $j$ has its own upper limit $\delta_j$ on the number of processors it can be assigned,
the model becomes basically equivalent to our shared-resource job scheduling problem (as discussed at the end of \cref{sec:preliminaries}).
\Textcite{drozdowski2001new} gave a simple greedy algorithm for minimizing the makespan in the offline setting (see also \cref{sec:minmakespan}).
\Textcite{decker200654} considered total completion time minimization for
%$n$
\emph{identical} malleable jobs
%on $m$ processors
for an otherwise rather general (possibly non-linear) speed-up function.
They gave a $5/4$-approximation for this setting.
\Textcite{beaumont2012minimizing} is closest to our model.
In particular, they assumed job-dependent resource limits $\delta_j$ that correspond to our resource requirements.
For minimizing weighted total completion time, they used a water-fill approach to prove the existence of structurally nice solutions (cf.~to the our water-filling approach in \cref{sec:minmakespan}).
Their main result is a (non-clairvoyant) $2$-approximation algorithm for the weighted case.
Their algorithm WDEQ assigns each job a number of processors according to their relative weight, but no more than the limit imposed by $\delta_j$.
Our results in \cref{sec:mintotalcompletiontime} yield an improved approximation ratio of $3/2+\varepsilon$ at the cost of clairvoyance (i.e., we must know the job's volumes and resource requirements).
Also, our algorithm only handles the unweighted case.

Other related models, such as rigid and moldable scheduling, disallow the resource assignment of a job to be adjusted after it has been started (see~\cite{DBLP:reference/crc/2004sch} for details).

%\paragraph{Continuously Shareable Resource Scheduling}

% \emph{Rigid} job scheduling refers to the classical scheduling variant where jobs require a fixed amount of resource (one or multiple machines) to be processed but are, otherwise, unaffected by the resource assignment.
% Since our focus lies on settings where the jobs' processing times are directly affected by the resource, we refer to~\cite{DBLP:reference/crc/2004sch} for details on these (and many more!) scheduling variants.

% In \emph{moldable} job scheduling, the jobs' processing times reduce (typically linearly) with the number of assigned machines, but that number must not change once the job is started.
% Related models can also be found under the name of scheduling with \emph{resource-dependent processing times}.
% We refer to \cite{jansen2002linear, mounie2007frac32, havill2008competitive} for moldable scheduling.

% \Textcite{jansen2002linear} give a PTAS for minimizing the makespan on a constant number of processors assuming arbitrary speedups an no preemptions.
% \Textcite{mounie2007frac32} derive a $(3/2+\varepsilon)$-approximation for minimizing the makespan for monotonic speedups without preemptions.
% If jobs have arrival times, \textcite{havill2008competitive} consider the non-preemptive case for linear speedups with an additional setup times to create/dispatch/destroy processes.
% They give a $4$-competitive algorithm for minimizing the makespan online.

\subsection{Our Contribution and Methods}%
\label{sec:contribution}

For our model, makespan minimization is known to be offline solvable (see \cref{sec:minmakespan}).
We thus concentrate on an online (list-scheduling) setting where jobs are given sequentially and we must commit to a resource assignment without knowing the number of jobs and future jobs' properties.
We use a water-filling approach that is known to produce \enquote{flattest} schedules \cite{beaumont2012minimizing}.
We derive properties that are necessary and sufficient for any $c$-competitive algorithm by providing conditions on \emph{c-extendable} schedules ($c$-competitive schedules to which we can add any job while remaining $c$-competitive).
From this, we derive slightly weaker \emph{universal schedules} that are just barely $c$-extendable and show that schedules derived via water-fill are always flatter than universal schedules.
Optimizing the value of $c$ yields $e/(e-1)$-competitiveness.
We then show that no algorithm can have a lower competitive ratio than $e/(e-1)$.

Our main result considers \emph{offline total completion time minimization}.
We improve upon the so far best result for this variant (a $2$-approximation~\cite{beaumont2012minimizing}) by providing a $(3/2+\varepsilon)$-approximation running polynomial time in $n,1/\varepsilon$.
The result relies on a continuous linear programming (CLP) formulation for the fractional total completion time, for which we consider primal-dual pairs.
The primal solution represents the resource assignments over time, while the dual represents the \emph{priority} of jobs over time.
We then extract additional properties about the primal/dual pair.
Roughly, our method is as follows.
We draw both the primal and dual solutions into a two-dimensional coordinate system.
See \Cref{fig:dual_line_schedule} for an example.
We then merge both solutions into a single 3D coordinate system by sharing the time axis and
use the these solutions as a blueprint for shapes in this coordinate system (see \Cref{fig:primal_dual_volumes}).
The volume of these shapes then correspond to parts of the primal and dual objective.
We use a second algorithm called \textsc{Greedy} that attempts to schedule jobs as early as possible.
Choosing the better one of the CLP and the greedy solution gives us the desired approximation.

\section{Makespan Minimization}%
\label{sec:minmakespan}

This \lcnamecref{sec:minmakespan} considers our resource-aware scheduling problem under the makespan objective.
For the offline problem, it is well-known that the optimal makespan $\OPTM(J)$ for a job set $J = \intcc{n}$ with total volume $V(J) = \sum_{j \in J} v_j$ is
\begin{math}
\OPTM(J)
=
\max\set{V(J)} \cup \set{p_j | j \in J}
\end{math}
and that a corresponding schedule can be computed in time $\ldauOmicron{n}$~\cite[Section~25.6]{DBLP:reference/crc/2004sch}.
The idea is to start with a (possibly infeasible) schedule $R$ that finishes all jobs at time $p_{\max} \coloneqq \max\set{p_j | j \in J}$ by setting $R_j(t) = v_j / p_{\max}$ for $t \in \intco{0, p_{\max}}$ and $R_j(t) = 0$ for $t > p_{\max}$.
This schedule uses a constant total resource of $\bar{R} \coloneqq V(J) / p_{\max}$ until all jobs are finished.
If $\bar{R} \leq 1$ (the resource is not overused), this schedule is feasible and optimal (any schedule needs time at least $p_{\max}$ to finish the \enquote{longest} job).
Otherwise we scale all jobs' resource assignments by $1/\bar{R}$ to get a new feasible schedule that uses a constant total resource of $1$ until all jobs are finished at time $V(J)$.
Again, this is optimal (any schedule needs time at least $V(J)$ to finish a total volume of $V(J)$).

\paragraph{List-Scheduling Setting}
Given that the offline problem is easy, the remainder of this \lcnamecref{sec:minmakespan} considers the (online) list-scheduling setting.
That is, an (online) algorithm $\cA$ receives the jobs from $J = \intcc{n}$ one after another.
Given job $j \in J$, $\cA$ must fix $j$'s resource assignment $R_j\colon \R_{\geq 0} \to \intcc{0, 1}$ without knowing $n$ or the properties of future jobs.
We refer to the resulting schedule by $\cA(J)$.
As usual in the online setting without full information, we seek to minimize the worst-case ratio between the costs of the computed and optimal schedules.
More formally, we say a schedule $R$ for a job set $J$ is \emph{$c$-competitive} if $M(R) \leq c \cdot \OPTM(J)$.
Similarly, we say an algorithm $\cA$ is \emph{$c$-competitive} if for any job set $J$ we have
\begin{math}
M\bigl( \cA(J) \bigr)
\leq c \cdot \OPTM(J)
\end{math}.

\paragraph{An Optimal List-Scheduling Algorithm}
Water-filling algorithms are natural greedy algorithms for scheduling problems with a continuous, preemptive character.
They often yield structurally nice schedules~\cite{beaumont2012minimizing, DBLP:conf/icpads/ChenWL09, DBLP:journals/tcs/AntoniadisKOR17}.
In this section, we show that water-filling (described below) yields a simple, optimal online algorithm for our problem.
\begin{theorem}%
\label{thm:minmakespan:optonline}
Algorithm \WFalg/ has competitive ratio $e / (e-1)$ for the makespan.
No deterministic online algorithm can have a lower worst-case competitive ratio.
\end{theorem}

We first describe a single step $\WFstep/(R, \iota, C)$ of \WFalg/ (illustrated in \cref{fig:WaterFill}).
It takes a schedule $R = \intoo{R_j}_{j \in J}$ for some job set $J$, a new job $\iota \notin J$, and a \emph{target completion time} $C$.
Its goal is to \emph{augment $R$ by $\iota$ with completion time $C$}, i.e., to feasibly complete $\iota$ by time $C$ without altering the resource assignments $R_j$ for any $j \in J$.
To this end, define the \emph{$h$-water-level}
\begin{math}
\wl_h(t)
\coloneqq
\min\set{r_{\iota}, \max\set{h - \bar{R}(t), 0}}
\end{math}
at time $t$ (the resource that can be assigned to $\iota$ at time $t$ without exceeding total resource $h$).
Note that $\iota$ can be completed by time $C$ iff
\begin{math}
\int_{0}^{C} \wl_1(t) \dif{t}
\geq
v_{\iota}
\end{math}
(the total leftover resource suffices to complete $\iota$'s volume by time $C$).
If $\iota$ cannot be completed by time $C$, $\WFstep/(R, \iota, C)$ \emph{fails}.
Otherwise, it \emph{succeeds} and returns a schedule that augments $R$ with the resource assignment $R_{\iota} = \wl_{h^*}$ for job $\iota$, where
\begin{math}
h^*
\coloneqq
\inf_{h \in \intcc{0, 1}}\set{h | \int_{0}^{C} \wl_h(t) \dif{t} \geq v_{\iota}}
\end{math}
is the smallest water level at which $\iota$ can be scheduled.

\WFalg/ is defined recursively via \WFstep/.
Given a job set $J = \intcc{n}$, define $H_j \coloneqq \OPTM(\intcc{j}) \cdot e/(e-1)$ as the target completion time for job $j \in J$ (remember that $\OPTM(\intcc{j})$ can be easily computed, as described at the beginning of this \lcnamecref{sec:minmakespan}).
Assuming \WFalg/ computed a feasible schedule $R^{(j-1)}$ for the first $j-1$ jobs (with $R^{(0)}(t) = 0$ $\forall t \in \R_{\geq0}$), we set $R^{(j)} \coloneqq \WFstep/(R^{(j-1)}, j, H_j)$.
If this step succeeds, the resulting schedule is clearly $e/(e-1)$-competitive by the choice of $H_j$.
The key part of the analysis is to show that indeed these water-filling steps always succeed.

We start the observation that water-fill schedules always result in \enquote{staircase-like} schedules (see \cref{fig:WaterFill}), a fact also stated in \cite{beaumont2012minimizing} (using a slightly different wording).
\begin{observation}[{\cite[Lemma~3]{beaumont2012minimizing}}]%
\label{lem:thm:minmakespan:monstaircase}
Consider a schedule $R$ whose total resource usage $\bar{R}$ is non-increasing (piecewise constant).
If we $\WFstep/(R, \iota, C)$ successfully augments $R$ by a job $\iota$, the resulting total resource usage is also non-increasing (piecewise constant).
\end{observation}

Next, we formalize that \WFstep/ generates the "flattest" schedules: if there is \emph{some} way to augment a schedule by a job that completes until time $C$, then the augmentation can be done via \WFstep/.
\begin{definition}
The \emph{upper resource distribution} $A^C_R(y)$ of a schedule $R$ is the total volume above height $y$ before time $C$ in $R$.
Given schedules $R,S$ (for possibly different job sets), we say $R$ is \emph{flatter} than $S$ ($R \preceq S$) if
\begin{math}
A^C_R(y)
\leq
A^C_S(y)
\end{math}
$\forall C \in \R_{\geq0},y \in \intcc{0, 1}$.
\end{definition}

\begin{lemma}[{\cite[Lemma~4, slightly generalized]{beaumont2012minimizing}}]%
\label{lem:WF:optimality}
Consider two schedules $R \preceq S$ for possibly different job sets.
Let $S'$ denote a valid schedule that augments $S$ by a new job $\iota$ completed until time $C$.
Then $\WFstep/(R, \iota, C)$ succeeds and
\begin{math}
\WFstep/(R, \iota, C)
\preceq
S'
\end{math}.
\end{lemma}

Next, we characterize $c$-competitive schedules that can be augmented by \emph{any} job while staying $c$-competitive.
%(this holds for any schedule computed by a $c$-competitive algorithm).
\begin{definition}
A schedule $R$ is \emph{$c$-extendable} if it is $c$-competitive and if it can be feasibly augmented by \emph{any} new job $\iota$ such that the resulting schedule is also $c$-competitive.
\end{definition}

\begin{lemma}%
\label{lem:extendability}
Consider a job set $J$ of volume $V$ and with maximal processing time $p_{\max}$.
A $c$-competitive schedule $R$ for $J$ is $c$-extendable if and only if
\begin{dmath}%
\label{eqn:extendability}
\forall y \text{ with } (c-1)/c < y \leq 1\colon
\quad
A^{\infty}_R(y)
\leq
(c-1) \cdot (1-y)/y \cdot \max\set{V, p_{\max} \cdot y}
\end{dmath}.
\end{lemma}

See \cref{app:proof:lem:extendability} for the proof of \cref{lem:extendability}.
While \cref{lem:extendability} gives a strong characterization, the bound on the right hand side of \cref{eqn:extendability} cannot be easily translated into a proper schedule for the given volume.
Thus we introduce proper (idealized) schedules that adhere to a slightly weaker version of \cref{eqn:extendability}.
These schedules are barely $e/(e-1)$-extendable.
Our proof of \cref{thm:minmakespan:optonline} combines their existence with \cref{lem:WF:optimality} to deduce that \WFalg/ is $e/(e-1)$-competitive.
\begin{definition}%
\label{def:universalschedules}
For any $V \in \R_{\geq0}$ we define the \emph{universal schedule}\footnote{
    One can think of $U_V$ as a schedule for a single job of volume $V$ and resource requirement $1$.
    Since there is only one job, we identify $U_V$ with its total resource requirement function $\bar{U}_V$.
} $U_V\colon \R_{\geq0} \to \intcc{0, 1}$ via
\begin{dmath}
U_V(t)
\coloneqq
\begin{cases}
1                                        & \text{if $t < \frac{1}{e-1} \cdot V$,}\\
1 - \ln\big( t \cdot \frac{c-1}{V} \big) & \text{if $\frac{1}{e-1} \cdot V \leq t < \frac{e}{e-1} \cdot V$, and}\\
0                                        & \text{otherwise.}
\end{cases}
\end{dmath}
\end{definition}
See \cref{fig:gt_incoming_job} for an illustration of universal schedules.
With $c = e / (e-1)$, one can easily check that
\begin{math}
A^{\infty}_{U_V}(y)
=
\frac{e^{1-y}-1}{e-1} \cdot V
\leq
(c-1) \cdot \frac{1-y}{y} \cdot V.
\end{math}
Thus, by \cref{lem:extendability}, universal schedules (and any flatter schedules for the same volume) are $e/(e-1)$-extendable.
Our final auxiliary \lcnamecref{lem:WF:maintainsextendability} extends the optimality of \WFalg/ from \cref{lem:WF:optimality} to certain augmentations of universal schedules.\footnote{
    \Cref{lem:WF:maintainsextendability} is not a special case of \cref{lem:WF:optimality}: the schedule $S'$ from \cref{lem:WF:optimality} must adhere to the new job's resource requirement, which is not the case for the universal schedule $U_{V + v}$.
}
See \cref{app:proof:lem:maintainsextendability} for the proof of \cref{lem:WF:maintainsextendability}.
\begin{lemma}%
\label{lem:WF:maintainsextendability}
Consider the universal schedule $U_V$, a new job $\iota$ of volume $v$ and processing time $p$, as well as a target completion time $H \geq \frac{e}{e-1} \cdot \max\set{V + v, p}$.
Then
\begin{math}
\WFstep/(U_V, \iota, H)
\preceq
U_{V + v}
\end{math}.
\end{lemma}

The above enables us to prove the competitiveness of \WFalg/ from \cref{thm:minmakespan:optonline}:
We show inductively that \WFalg/ produces a feasible schedule $R^{(j)}$ for the first $j$ jobs (using that $R^{(j-1)}$ is \enquote{flatter} than $U_{V(\intcc{j-1})}$ together with \cref{lem:WF:optimality}) and use this to prove $R^{(j)} \preceq U_{V(\intcc{j})}$ (via \cref{lem:WF:maintainsextendability}).
By universality, this implies that all $R^{(j)}$ are $e/(e-1)$-extendable (and thus, in particular, $e/(e-1)$-competitive).
The full prove of \WFalg/ is given in \cref{app:wfoptimality}.

\section{Total Completion Time Minimization}%
\label{sec:mintotalcompletiontime}

This \lcnamecref{sec:mintotalcompletiontime} considers the total completion time minimization and represents our main contribution.
In contrast to offline makespan minimization (\Cref{sec:minmakespan}), it remains unknown whether there is an efficient algorithm to compute an offline schedule with minimal total completion time.
The so far best polynomial-time algorithm achieved a $2$-approximation~\cite{beaumont2012minimizing}.
% \peter{
%     Can we very briefly summarize the idea of that $2$-approximation?
%     (Or maybe better not here but in "Our contribution/Methods"?)
% }
%Christoph: has been done in related work.
We improve upon this, as stated in the following \lcnamecref{thm:three_halfs_plus_eps}.
\begin{restatable}{theorem}{thmthreehalfspluseps}
    \label{thm:three_halfs_plus_eps}
    There is a $(3/2+\varepsilon)$-approximation algorithm for total completion time minimization.
    Its running time in polynomial time in $n$ and $1/\varepsilon$.
\end{restatable}

For clarity of presentation, we analyze an idealized setting in the main part.
The details for the actual result can be found in \Cref{sec:threehalfspluseps, app:details_for_section_3}.

% \peter[inlinepar]{
%     Give a brief overview of the remaining section.
%     Mention already here that for clarity of presentation we analyze slightly "idealized" setting in the main part and give details for the actual result in the appendix.
% }

\paragraph{Algorithm Description}
\label{subsec:mainpart_algorithm_description}

Our algorithm computes two \emph{candidate schedules} using the two sub-algorithms \textsc{Greedy} and \textsc{LSApprox} (described below).
It then returns the schedule with smallest total completion time among both candidates.
%It then compares the resulting schedules and simply returns the schedule that has the smallest total completion time among both candidates.

Sub-algorithm \textsc{Greedy} processes the jobs in ascending order of their volume.
To process a job, \textsc{Greedy} assigns it as much resource as possible as early as possible in the schedule.
Formally,
for jobs $J=[n]$ ordered as $v_1\le \dots \le v_n$, the schedule $R^G$ for \textsc{Greedy} is calculated recursively using
$R^G_j(t)=\mathds{1}_{t<t_j}\cdot\min(r_j,1-\sum_{i=1}^{j-1} R^G_{i}(t))$\footnote{$\mathds{1}_{t<t_j}$ denotes the indicator function that is $1$ for $t<t_j$ and $0$ everywhere else.}, where the completion time $t_j$ for job $j$ is set such that $j$ schedules exactly its volume $v_j$.
See \Cref{fig:greedy_schedule_example} for an example of a \textsc{Greedy} schedule.
Sub-algorithm \textsc{LSApprox} deals with solutions to following \emph{continuous linear program} ($CLP$).
%Figures created using Ipe extensible drawing editor.
\begin{figure}[t]
    \begin{subfigure}{0.49\linewidth}
        \includegraphics[width=1.0\linewidth]{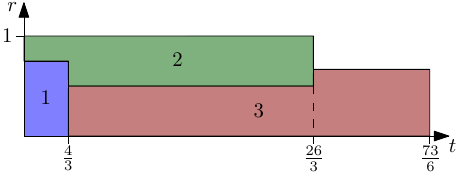}
        \caption {
        }
        \label{fig:greedy_schedule_example}
    \end{subfigure}
    \begin{subfigure}{0.49\linewidth}
        \includegraphics[width=1.0\linewidth]{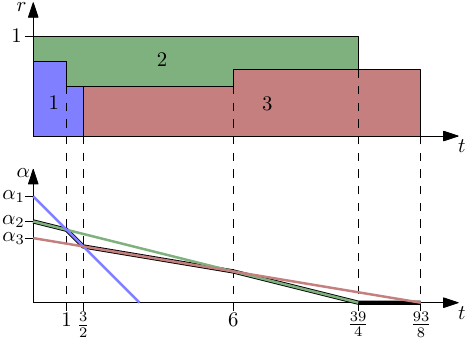}
        \caption{
        }
        \label{fig:dual_line_schedule}
    \end{subfigure}
    \caption{
        Schedules for a job set $J=[3]$ with \textcolor{j1_color}{$(v_1,r_1)=(1,3/4)$}, \textcolor{j2_color}{$(v_2,r_2)=(4,1/2)$} and \textcolor{j3_color}{$(v_3,r_3)=(6,2/3)$}.
        (a) \textsc{Greedy}'s schedule, 
        (b) Above: A primal (resource) schedule.
        Below: A dual (priority) schedule.
        With the dual variables having values $\alpha_1=51/16$,$\alpha_2=39/16$ and $\alpha_3=31/16$,
        the volumes of the jobs are exactly scheduled.
        (See \Cref{calc:fig:dual_line_schedule} for calculations.)
        \vspace{-0.5em}
    }
    \label{fig:greedy_ls}
\end{figure}
\label{def:CLP}
\begin{align*}
    \text{minimize } &\sum_{j\in J}\int_0^\infty{\frac{t\cdot R_j(t)}{v_j}\dt}~~~~~~~~&
    \int_0^\infty{R_j(t)\dt}\ge v_j~~\forall j\in J\\
    &0\le R_j(t)\le r_j~~\forall j\in J, t\in \mathbb{R}_{\ge 0}~~~~~~~~&
    \sum\nolimits_{j\in J} R_{j}(t)\le 1~~\forall t\in \mathbb{R}_{\ge 0}
    %\text{\footnoteref{ftn:regularity}}
\end{align*}
Roughly, \textsc{LSApprox} first subdivides the job set into those jobs that produce a high completion time and the remaining jobs. For the former,
an approximate solution is computed using the dual to the discretization (an LP) of above $CLP$.
For the latter, is is enough to reserve a small portion of the resource to schedule them with small completion times.
The details of this algorithm are explained in \Cref{subsec:algorithm_desc_three_halfs_plus_eps}.

For clarity of presentation, the main part will only do a simplified analysis using an idealization of \textsc{LSApprox}.
For the detailed analysis using \textsc{LSApprox}, we refer to \Cref{sec:threehalfspluseps}.
For the analysis of \textsc{Greedy}, we refer to \Cref{app:greedy}.

%\medskip

%We approximate the $CLP$ because we are unable to solve it exactly.
%We do, however, show that we can replace an optimal solution to the $CLP$ with an approximate one that we actually \emph{can} compute.
%This approximation is used to obtain the main result for total completion time minimization (see \Cref{sec:threehalfspluseps}).
% \peter{
%     Question that arises here: Could we simply combine Greedy with an arbitrary $c$-approximation for fractional total completion time and get a sensible approximation for (integral) total completion time?
% }

\subsection{Analysis via a Bounded Fractionality Gap}
%\peter{incorporate this paragraph into previous part of section}

Throughout the analysis, we use $C^*$ to denote the optimal total completion time and $C^{F*}$ for the optimal fractional total completion time.
We require an algorithm that produces a schedule $R$ with a small \emph{fractionality gap} $\gamma(R)\coloneqq C(R)/C^{F*}$, i.e., we compare the total completion time of $R$ with the optimal fractional total completion time for the same job set.
We show the following generalization of \Cref{thm:three_halfs_plus_eps}.

\begin{restatable}{theorem}{thmthreehalfs}
    \label{lem:approx}
    Assume that there is a polynomial-time algorithm $A$ for total completion time minimization that produces a schedule $R$ with $\gamma(R)\ge 1$.
    Then there exists a polynomial-time $(\gamma(R)+1)/2$-approximation for total completion time minimization.
\end{restatable}

The proof of \Cref{lem:approx} relies on \Cref{prop:lower_bounds} (three lower bounds on the optimal total completion time) and \Cref{prop:greedy_bound} (\textsc{Greedy}'s objective in relation to these bounds).
Lower Bound~(1) (\emph{Squashed Area} Bound) and Bound~(2) (\emph{Length} or \emph{Height} Bound) are due to \textcite[Def. 6,7]{beaumont2012minimizing}. Bound (3) is our novel lower bound.
The proof can be found in \Cref{app:lowerbound12}.

\begin{restatable}{proposition}{proplowerbounds}
    \label{prop:lower_bounds}
    Assuming $v_1\le \dots\le v_n$, the following are lower bounds on $C^*$:
    \vspace{-0.2em}
    \begin{multicols}{3}
        \begin{enumerate}[nosep, leftmargin=2.7em]
            \item[(1)] {
                $C^L\coloneqq \max_{j\in J} p_j$
            }
            \item[(2)] {
                $C^A\coloneqq \sum_{j=1}^n \sum_{i=1}^{j} v_j$
            }
            \item[(3)] {
                $C^{F*}+1/2\cdot C^L$
            }
        \end{enumerate}
    \end{multicols}
\end{restatable}

\vspace{-0.8em}

\begin{restatable}{proposition}{greedybound}
    \label{prop:greedy_bound}
    The \textsc{Greedy} schedule $R^G$ satisfies $C(R^G)\le C^A+C^L$.
\end{restatable}

Using them, we can give the proof of \Cref{lem:approx}.

%When considered in isolation, \Cref{prop:greedy_bound} and \Cref{prop:fractional_bound} each provide $2$-approximations already.
%However, by combining both, we can prove a $3/2$-approximation (\Cref{lem:approx}):

\begin{proof}[Proof of \Cref{lem:approx}]
    We run both \textsc{Greedy} and $A$ in polynomial time to produce schedules $R^G$ and $R^A$, respectively,
    and choose the schedule with the smaller total completion time.
    Using \Cref{prop:lower_bounds, prop:greedy_bound} and the fractionality gap $\gamma\coloneqq\gamma(R^A)$, we can bound the cost $C\coloneqq\min(C(R^A),C(R^G))$
    of the resulting schedule in terms of $C^*$:
    %hspaces to get both correct alignment and the qed symbol at the correct place
    \begin{equation*}
        \hspace{-8.6em}
        C
        \le \min(\gamma\cdot C^{F*},C^A+C^L)
        \le \min(\gamma\cdot (C^*-1/2\cdot C^L),C^*+C^L)
    \end{equation*}
    \begin{equation*}
        =\frac{\gamma+1}{2} C^*- \frac{\gamma+2}{4} C^L + \min\left(\frac{\gamma-1}{2} C^*-\frac{\gamma+2}{4} C^L,\frac{\gamma+2}{4}C^L-\frac{\gamma-1}{2}C^*\right)\le \frac{\gamma+1}{2}C^*\qedhere
    \end{equation*}
\end{proof}

\subsection{The fractionality gap of line schedules}
\label{subsec:fractional}

For the remainder of this paper, we will introduce \emph{line schedules} and their structural properties.
Roughly, a line schedule is a certain primal-dual pair for the $CLP$ defined in \Cref{def:CLP}, and its dual, which we call $DCP$:
\label{def:DCP}
\begin{align*}
    \text{maximize } &\sum\nolimits_{j\in J} \alpha_j v_j - \sum\nolimits_{j\in J}r_j \int_{0}^{\infty} \beta_{j}(t) \dt-\int_{0}^{\infty}\gamma(t)\dt\\
    \text{s.t. }&\alpha_j,\beta_j(t),\gamma(t)\ge 0~~\forall j\in J,t\in \mathbb{R}_{\ge 0}~~~~~~~
    \gamma(t)+\beta_j(t)\ge \alpha_j-t/v_j~~\forall j\in J,t\in \mathbb{R}_{\ge 0}
\end{align*}
It is obtained by dualizing the time-discretized version of the $CLP$ (see \Cref{par:time_discretized_lp_and_its_dual}) and then extending its constraints to the continuous time domain.
\emph{Line schedules} formalize the idea that, if we know the dual $\alpha$-values, we can reconstruct all remaining primal/dual variables to obtain a primal-dual pair.
If the $\alpha$-values are chosen correctly, then the volumes scheduled in the primal are exactly the desired volumes $(v_j)_{j\in J}$.

To this end, we will \emph{assume} that we have access to an algorithm called \textsc{LS} that produces such a line schedule $R^F$ with $C^F(R^F)=C^{F*}$.
%An example is shown in \Cref{fig:dual_line_schedule}.
We can then show that \textsc{LS} produces schedules with a fractionality gap of $2$:

\begin{restatable}{proposition}{fractionalbound}
    \label{prop:fractional_bound}
    The \textsc{LS} schedule $R^F$ satisfies $\gamma(R^F)\le 2$.
\end{restatable}

In the following, we develop the details of line schedules.
To this end, first define \emph{primal-dual pair}
as a tuple $(R,\alpha,\beta,\gamma,v)$ that fulfills the following continuous \emph{slackness conditions (sc)}. Again, these are found by extending the time-discretized version of the $CLP$ to the continuous time domain. These conditions hold for all $j\in J$ and $t\in\R_{\ge 0}$.
\begin{flalign*}
    &\text{($\alpha$-sc): }\alpha_j(\bar{v}_j-\int_{0}^{\infty}{R_j(t)\dt})=0
    &&\text{($\beta$-sc): }\beta_j(t)(r_j-R_j(t))=0~~
    \\
    &\text{($\gamma$-sc): }\gamma(t)(1-\sum\nolimits_{j\in J} R_j(t))=0~~
    &&\text{($R$-sc): }R_j(t)(\alpha_j-t/v_j-\beta_j(t)-\gamma(t))=0~~
\end{flalign*}

If we choose arbitrary $\alpha$-values, then the corresponding line schedule is still a primal-dual pair, except that it possibly schedules a different set of volumes, i.e., the $\alpha$-sc is only true if we replace $v_j$ in the constraint by some other volume $\bar{v}_j$.
This fact is used for the detailed proof of our $(3/2+\varepsilon)$-approximation, see \Cref{subsec:generalized_clp_dcp}.

To this end, define the \emph{dual line} $d_j(t)\coloneqq\alpha_j-t/v_j$ for each $j\in J$.
The intuition behind a line schedule is now that the heights of the dual lines represent priorities:
Jobs are scheduled (with maximum remaining schedulable resource) in decreasing order
of the dual line heights at the respective time points.
Jobs are not scheduled if their dual line lies below zero.
This is formalized in the following \lcnamecref{def:line_schedule}.
(In \Cref{fig:dual_line_schedule}, we supplement the example from \Cref{fig:greedy_schedule_example} by a depiction of the dual lines.)

\begin{definition}
    \label{def:line_schedule}
    We call a job set $J$ \emph{non-degenerate} if all job volumes are pairwise distinct, i.e., $v_j\ne v_{j'}$ for all $j,j'\in J$.\footnote{While not strictly required, this makes line schedules unique and simplifies the analysis.}
    Define a total order for each $t\ge 0$ as $j'\succ_t j:\Leftrightarrow d_{j'}(t)>d_j(t)\text{ or } d_{j'}(t)=d_j(t)\text{ and } v_{j'}>v_j$.\footnote{The second part of the definition ($d_{j'}(t)=d_j(t)$ and $v_{j'}>v_j$) only exists for disambiguation of the line schedule, but is not further relevant.}
    The \emph{line schedule} of $\alpha$ is a tuple
    $(R,\alpha,\beta,\gamma,v)$ (recursively) defined as follows.
    \begin{align*}
        R_j(t)&=\mathds{1}_{d_j(t)>0}\cdot \min(r_j,1-\sum\nolimits_{j'\succ_t j} R_{j'}(t))~~~~~~~~~~\beta_j(t)=\max(0,d_j(t)-\gamma(t))
        \\
        \gamma(t)&=\max(0,d_j(t))\text{, where $j$ is the smallest job according to $\succ_t$ with $R_j(t)>0$}
    \end{align*}
\end{definition}

Equipped with the definition of a line schedule, we can now tackle the proof of
\Cref{prop:fractional_bound}.
It requires the following two properties about the assumed algorithm \textsc{LS}.
First, \Cref{lem:alphajvj_are_completion_times} allows us to bound the completion times of a fractional schedule in terms of the $\alpha$-variables in the $DCP$:

\begin{restatable}{lemma}{lemalphajvjarecompletiontimes}
    \label{lem:alphajvj_are_completion_times}
    Algorithm \textsc{LS} produces a schedule $R^F$
    %such that there exists
    %a primal-dual pair
    %$(R^F,\alpha,\beta,\gamma)$
    with $C_j(R^F)\le \alpha_j v_j$ for all $j\in J$.
\end{restatable}

Second, we show the following \lcnamecref{lem:duality_balancedness}. Abbreviate $\Rho=\sum_{j\in J}\int_0^\infty{t\cdot R_j(t)/v_j\dt}$ (the primals objective), and
$\Alpha=\sum_{j\in J} \alpha_j v_j$, $\Beta=\sum_{j\in J}r_j \int_{0}^{\infty} \beta_{j}(t)\dt$ and $\Gamma=\int_{0}^{\infty}\gamma(t)\dt$ (the parts of the dual objective).

\begin{restatable}{lemma}{lemdualitybalancedness}
    \label{lem:duality_balancedness}
    Algorithm \textsc{LS} produces a schedule $R^F$ such that there exists
    a primal-dual pair $(R^F,\cdot,\cdot,\cdot)$ that fulfills strong duality ($\Alpha=\Beta+\Gamma+\Rho$) and balancedness ($\Rho=\Beta+\Gamma$).
\end{restatable}

Using these lemmas, we can show \Cref{prop:fractional_bound}.

\begin{proof}[Proof of \Cref{prop:fractional_bound}]
    Using \Cref{lem:alphajvj_are_completion_times, lem:duality_balancedness}, we show the statement as follows:
    \begin{equation*}
        C(R^F)
        =\sum_{j\in J} C_j(R^F)
        \le \sum_{j\in J}{\alpha_j v_j}
        =\Alpha
        =\Alpha-\Beta-\Gamma+\Rho
        =2\Rho
        =2C^F(R^F)
        =2C^{F*}\qedhere
    \end{equation*}
\end{proof}

In \Cref{subsec:lem_exact_volume_completion}, we show the following \Cref{lem:exact_volume_completion}, stating that line schedules are indeed primal-dual pairs.
We then define \textsc{LS} to output a schedule $R^F$ for a line schedule $(R^F,\alpha,\beta,\gamma,v)$ according to \Cref{lem:exact_volume_completion}, i.e., for each $j\in J$,
$\int_{0}^\infty R_j(t)\dt=v_j$.
Using this definition, we can show \Cref{lem:alphajvj_are_completion_times}.

\begin{restatable}{lemma}{lemexactvolumecompletion}
    \label{lem:exact_volume_completion}
    For any job set $J$ there exists an $\alpha$ such that the line schedule of $\alpha$ is a primal-dual pair.
\end{restatable}

%\lemalphajvjarecompletiontimes*

\begin{proof}[Proof of \Cref{lem:alphajvj_are_completion_times}]
    By definition,
    $R^F_j(t)=0$ if $d_j(t)\le 0$.
    Hence, as $d_j$ is monotonically decreasing,
    $C_j(R^F)$ is bounded by the zero of $d_j(t)$, which lies at $t=\alpha_j v_j$.
\end{proof}

%Figures created using Ipe extensible drawing editor.
\begin{figure}[t]
    \begin{subfigure}{0.5\textwidth}
        \centering
        \includegraphics[width=0.9\linewidth]{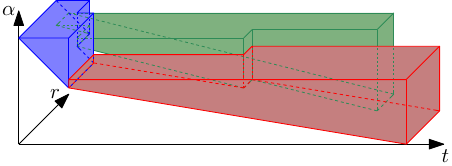}
        \caption {
        }
        \label{fig:dual_line_schedule_3D_primal}
    \end{subfigure}%
    \begin{subfigure}{0.5\textwidth}
        \centering
        \includegraphics[width=0.9\linewidth]{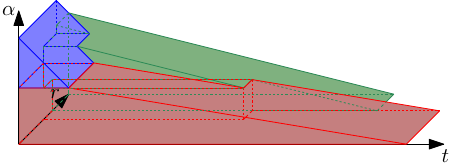}
        \caption {
        }
        \label{fig:dual_line_schedule_3D_dual}
    \end{subfigure}
    \caption {
        (a) $\Rho$-shapes for job set from \Cref{fig:dual_line_schedule}. $\Rho$-shapes are delimited from below by $d_j(t)$ (extended into the resource axis), from above by $\alpha_j$, and their top surface follows the primal schedule. (b) The shapes shown represent the union of $\Beta$- and $\Gamma$-shapes.
        They are delimited from the left (right) by $t=0$ ($d_j(t)$) (extended into the resource axis), and from top and bottom by the value of $d_j(t)$ at the starting and finishing time of some piece of $j$.
        See \Cref{def:omega_areas, def:shapes} for the formal definition of these shapes.
        \vspace{-0.5em}
    }
    \label{fig:primal_dual_volumes}
\end{figure}

The remainder of this section will initiate the proof of \Cref{lem:duality_balancedness}.
We first give a geometric understanding of the involved quantities ($\Rho,\Alpha,\Beta,\Gamma$).
We build a 3D coordinate system from a line schedule.
The time axis is shared, and the ordinates form the remaining two axes.
We then draw 3D shapes into this coordinate system that correspond to parts of the above quantities and therefore of the $CLP/DCP$ objectives. These shapes are described in detail in \Cref{subsec:geometric_shapes}.
Generally, these shapes are constructed such that the primal and dual schedules can be
\enquote{seen} from above or front.
In our case, the primal schedule will be seen from the top, and the dual schedule from the front.
\Cref{fig:primal_dual_volumes} illustrates the shapes in our construction.
For each part of the objective $\Psi\in\set{\Rho,\Alpha,\Beta,\Gamma}$, we have a corresponding shape
$\Psi^{\mathrm{all}}$, which is subdivided into pieces $\Psi^{i,l}$, respectively.

We can show that certain pieces are pairwise non-overlapping (\Cref{lem:no_overlap}), that the $\Alpha$-pieces make up all other pieces (\Cref{lem:alpha_shapes_composed_of_other_shapes}) and we can relate the volume of these pieces with one another and with the actual objective (\Cref{lem:shape_equalities}).

\begin{lemma}
    \label{lem:no_overlap}
    Let $V$ and $W$, $V\ne W$, be $\Rho$-pieces, $\Beta$-pieces or $\Gamma$-pieces (every combination allowed), or both be $\Alpha$-pieces.
    Then $V$ and $W$ do not overlap.
\end{lemma}

\begin{lemma}
    \label{lem:alpha_shapes_composed_of_other_shapes}
    $\Alpha^{\mathrm{all}}$ is composed of the other shapes, i.e.,
    $\Alpha^{\mathrm{all}}=\Rho^{\mathrm{all}}\cup \Beta^{\mathrm{all}}\cup\Gamma^{\mathrm{all}}$.
\end{lemma}

\begin{lemma}
    \label{lem:shape_equalities}
    The following statements hold:
    \begin{multicols}{2}
        \begin{enumerate}[nosep, leftmargin=*]
            \item {
                $\abs{\Rho^{i,l}}=\abs{\Beta^{i,l}}+\abs{\Gamma^{i,l}}$ for all $i,l$.
            }
            \item {
                $\abs{\Psi^{\mathrm{all}}}=\Psi$ for all $\Psi\in\set{\Rho,\Alpha,\Beta,\Gamma}$
            }
        \end{enumerate}
    \end{multicols}
\end{lemma}

Due to space limitations, we give the actual construction of the pieces and the proofs of \Cref{lem:no_overlap, lem:alpha_shapes_composed_of_other_shapes, lem:shape_equalities} in \Cref{subsec:geometric_shapes}.
Now we can give the proof of \Cref{lem:duality_balancedness}.

\begin{proof}[Proof of \Cref{lem:duality_balancedness}]
    Using \Cref{lem:no_overlap, lem:alpha_shapes_composed_of_other_shapes, lem:shape_equalities}, we get
    \begin{equation*}
        \hspace{-5.0em}
        \Alpha
        =\abs{\Alpha^{\mathrm{all}}}
        =\abs{\Rho^{\mathrm{all}}\cup \Beta^{\mathrm{all}}\cup \Gamma^{\mathrm{all}}}
        =\abs{\Rho^{\mathrm{all}}}+\abs{\Beta^{\mathrm{all}}}+\abs{\Gamma^{\mathrm{all}}}
        =\Rho+\Beta+\Gamma
    \end{equation*}
    \begin{equation*}
        \hspace{-1.0em}
        =\abs{\Rho^{\mathrm{all}}}+\abs{\Beta^{\mathrm{all}}}+\abs{\Gamma^{\mathrm{all}}}
        =\sum\nolimits_{i,l} \abs{\Rho^{i,l}}+\abs{\Beta^{i,l}}+\abs{\Gamma^{i,l}}
        =\sum\nolimits_{i,l} 2\abs{\Rho^{i,l}}
        =2\Rho
        .\qedhere
    \end{equation*}
\end{proof}

%from author guidelines:
%"If a paper includes an Appendix, it should be placed in front of the references."

\clearpage
\appendix
\makeatletter
\@mkboth{APPENDIX}{APPENDIX}
\makeatother

\section{Details for \Cref{sec:minmakespan}}%
\label{app:minmakespan}

\subsection{Proof of \cref{thm:minmakespan:optonline}}%
\label{app:wfoptimality}

We first prove that \WFalg/ is competitive.
Consider an arbitrary job set $J = \intcc{n}$.
For $j \in J \cup \set{0}$ let $R^{(j)}$ denote the schedule \WFalg/ produces for the job set $\intcc{j} \subseteq J$ of the first $j$ jobs.
Similarly define $O^{(j)} \coloneqq \OPTM(\intcc{j})$ as the optimal makespan for the first $j$ jobs.
To simplify notation, we define $V(j) \coloneqq V(\intcc{j})$ as the total volume of the first $j$ jobs.

We show inductively that the computation of $R^{(j)}$ succeeds and that $R^{(j)} \preceq U_{V(j)}$ for all $j \in J \cup \set{0}$.
By universality, this implies that all $R^{(j)}$ are $e/(e-1)$-extendable (and thus, in particular, $e/(e-1)$-competitive).

For the base case $j = 0$, note that $R^{(0)}$ and $U_{V(0)}$ are identical (the trivial $0$-schedule that schedules no volume at all).
Thus, clearly $R^{(0)} \preceq U_{V(0)}$.
Now consider $j \geq 1$ and assume $R^{(j-1)} \preceq U_{V(j-1)}$.
By definition of universal schedules, $U_{V(j-1)}$ can be feasibly augmented by $j$ with completion time $H_j = \frac{e}{e-1} \cdot O^{(j)} \geq \frac{e}{e-1} \cdot \max\set{V(j), p_j}$.
Using \cref{lem:WF:optimality}, we get
\begin{math}
R^{(j)}
=
\WFstep/(R^{(j-1)}, j, H_j)
\preceq
\WFstep/(U_{V(j-1)}, j, H_j)
\end{math}.
Combining this with \cref{lem:WF:maintainsextendability}, which gives
\begin{math}
\WFstep/(U_{V(j-1)}, j, H_j)
\preceq
U_{V(j)}
\end{math},
we get the desired statement
\begin{math}
R^{(j)}
\preceq
U_{V(j)}
\end{math},
finishing the proof of the competitiveness.

\medskip

For the optimality of \WFalg/, consider the algorithm $c$-\WFalg/ that is identical to \WFalg/ but uses target completion times $H_j = c \cdot \OPTM(\intcc{j})$ in the recursion.
By \cref{lem:WF:optimality}, if there is a $c$-competitive algorithm, then $c$-\WFalg/ cannot fail, as it can schedule the jobs with the same completion times $C_j \leq c \cdot \OPTM(\intcc{j})$.
Thus, it is sufficient to prove that for any $c < e/(e-1)$ there is an instance for which $c$-\WFalg/ fails.
% \christoph{The argument is actually independent of \WFalg/}

Let $c < e/(e-1)$ and consider the job set $J = \intcc{n}$ with $v_j \coloneqq 1/n$ and $r_j \coloneqq 1/j$ for $j \in J$.
By construction, $\OPTM(\intcc{j}) = j/n$.
Thus, the target completion times for $c$-\WFalg/ are $H_j = c \cdot j/n$.

For the sake of a contradiction, assume $c$-\WFalg/ successfully computes a feasible schedule for $J$.
Consider the intervals
\begin{math}
I_j
\coloneqq
\intco{H_j, H_{j+1}}
=
\intco{c \cdot \alpha, c \cdot \alpha + c/n}
\end{math}
for $\alpha \coloneqq j/n$ and note that $c$-\WFalg/ cannot schedule any job $j' < j$ during $I_j$ (by construction, $c$-\WFalg/ completes $j'$ at time $H_{j'} < H_j$).
Thus, the total resource usage during $I_j$ is at most $\min\set{\frac{1}{j} + \frac{1}{j+1} + \dots + \frac{1}{n}, 1} = \min\set{\cH_n - \cH_{j-1}, 1}$, where $\mathcal{H}_k$ denotes the $k$-th harmonic number.
For $n \to \infty$, $I_j = I_{\alpha n}$ becomes a point interval at time point $H_j = c \alpha$ with total resource usage at most $\min\set{\lim_{n \to \infty} \cH_n - \cH_{\alpha n-1}, 1} = \min\set{\lim_{n \to \infty} \cH_n - \cH_{\alpha n}, 1}$.
Using the Euler-Mascheroni constant $\gamma \coloneqq \lim_{n \to \infty} ( \cH_n - \ln{n} )$, we evaluate
\begin{equation}
\begin{aligned}
&
\lim_{n\to\infty}{\cH_n-\cH_{\alpha n}}
=
\lim_{n\to\infty}{(\cH_n-\ln{n})+\ln{n}-(\cH_{\alpha n}-\ln{\alpha n})-\ln{\alpha n}}
\\{}={}&
\lim_{n\to\infty}{\gamma+\ln{n}-\gamma-\ln{\alpha n}}
=
-\ln{\alpha}
\end{aligned}
\end{equation}

Thus, at time point $t = \alpha \cdot c$, the algorithm has total resource usage at most $-\ln(t/c)$.
This is non-increasing in $t$, so we look for the time point $t^*$ where $-\ln(t^*/c) = 1$ (before $t^*$ the schedule has total resource usage $1$).
Solving this yields $t^* = c / e$.
Since $c$-\WFalg/ must be finished by $H_n = c$, the total volume it can schedule is at most
\begin{equation}
t^*\cdot 1 + \int_{t^*}^{c}-\ln{\frac{t}{c}} \dt
=
c e^{-1} + \int_{c e^{-1}}^{c}-\ln{\frac{t}{c}} \dt
=
c e^{-1} + \left[t-t\ln{\frac{t}{c}}\right]_{c e^{-1}}^c
=
c \frac{e-1}{e}
\end{equation}
Because of $c < e/(e-1)$ and for large enough $n$, this implies that $c$-\WFalg/ schedules a total volume of strictly less than $1$ (the total volume of the job set), contradicting the feasibility of the computed schedule.

\subsection{Proof of \Cref{lem:extendability}}%
\label{app:proof:lem:extendability}

Consider a $c$-competitive schedule $R$ for a job set $J$ of volume $V$ and with maximal processing time $p_{\max}$.
Let $\OPT \coloneqq \max\set{V, p_{\max}}$ denote the optimal makespan for $J$ and set $H \coloneqq c \cdot \OPT$.
Schedule $R$ is $c$-extendable if and only if it can be augmented by \emph{any} new job of volume $v \in \R_{\geq0}$, resource requirement $r \in \intoc{0, 1}$  and processing time $p \coloneqq v/r$ that completes until time $H' \coloneqq c \cdot \OPT'$, with $\OPT' \coloneqq \max\set{V + v, p_{\max}, p}$.
Let $A_R \coloneqq A^{\infty}_R$.
Note that the new job can be scheduled with completion time $H'$ if and only if
\begin{dmath}
r \cdot H' - A_R(1 - r)
\geq
v
\end{dmath}
(the available free area before $H'$ suffices to schedule volume $v$).
Rearranging and using the definitions of $H'$ we get
\begin{dmath}
\begin{aligned}
A_R(1 - r)
&\leq
r \cdot \intoo[\big]{c \cdot \max\set{V + p \cdot r, p_{\max}, p} - p}\\
&=
r \cdot \max\set{c \cdot V + (c r - 1) \cdot p, c \cdot p_{\max} - p, (c - 1) \cdot p}
.
\end{aligned}
\end{dmath}
Note that for $r \geq 1/c$, the right-hand side is at least $c r \cdot V + (c r - 1) \cdot p r \geq V$, while the left-hand side is clearly at most $V$.
Thus, in this case the inequality is trivially true.

So assume $r < 1/c$.
Note that the left-hand side does not depend on $p$.
For the right hand side, we can compute the partial derivatives w.r.t.~$p$ of the three terms in the maximum.
Note that only the rightmost term in the maximum has a positive derivative and that for $p=0$ it is zero (and, thus, clearly smaller than the other two terms).
This implies that the worst case (minimal right-hand side) occurs for $p$ such that
\begin{math}
(c-1) \cdot p
=
\max\set{c \cdot V + (c r - 1) \cdot p, c \cdot p_{\max} - p}
\end{math},
which is equivalent to
\begin{math}
p
=
\max\set{V + r \cdot p, p_{\max}}
\end{math}.
Using that $p = V + r \cdot p$ implies (by recursion) $p = V / (1-r)$, the worst-case for our inequality becomes
\begin{dmath}
\begin{aligned}
A_R(1 - r)
\leq
r \cdot \max\set{c \cdot V + (c r - 1) \cdot p, c \cdot p_{\max} - p}
=
(c-1) \cdot r/(1-r) \cdot \max\set{V, p_{\max} \cdot (1-r)}
.
\end{aligned}
\end{dmath}
Using the substitution $y = 1-r$, we get the desired result.

\subsection{Proof of \cref{lem:WF:maintainsextendability}}%
\label{app:proof:lem:maintainsextendability}

We first note that the condition $H \geq \frac{e}{e-1} \cdot \max\set{V + v, p}$ on the target completion time ensures that $\WFstep/(U_V, \iota, H)$ succeeds:
The value $\max\set{V + v, p}$ is the optimal makespan for scheduling $\iota$ together with the volume from $U_V$ (remember that $U_V$ can be thought of as scheduling a single job of volume $V$ and resource requirement $1$).
Since universal schedules are $e/(e-1)$-extendable, $U_V$ can be augmented by $\iota$ finishing at time $\frac{e}{e-1} \cdot \max\set{V + v, p}$ (or later).
Thus, by \cref{lem:WF:optimality}, $\WFstep/(U_V, \iota, H)$ succeeds.

Now let the $W$ and $U$ denote the schedules $\WFstep/(U_V, \iota, H)$ and $U_{V + v}$, respectively.
To simplify notation, we also identify $W$ and $U$ with their total resource usage functions $\bar{W}$ and $\bar{U}$, respectively.
We must prove $W \preceq U$, which is equivalent to showing that
\begin{math}
\Delta(C, y) \coloneqq A^C_U(y) - A^C_W(y)
\geq
0
\end{math}
for all $C \in \R_{\geq0}$ and $y \in \intcc{0, 1}$.

%Figures created using Ipe extensible drawing editor.
\begin{figure}
\begin{subfigure}{0.5\linewidth}
\centering\includegraphics[width=\linewidth]{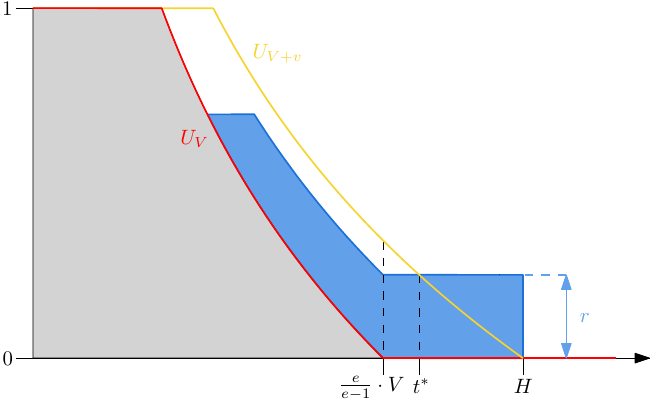}
\caption{
    $U_{V+v}$ intersects $W$ at the lower plateau.
}\label{fig:gt_incoming_job:a}
\end{subfigure}
\begin{subfigure}{0.5\linewidth}
\centering\includegraphics[width=\linewidth]{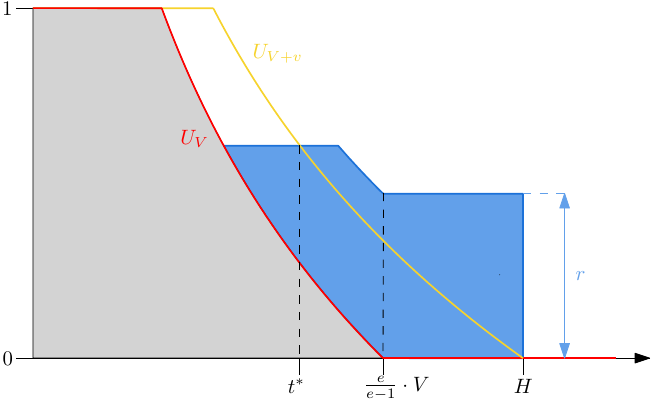}
\caption{
    $U_{V+v}$ intersects $W$ at the upper plateau.
}\label{fig:gt_incoming_job:b}
\end{subfigure}
\caption{
    Universal schedules $U_V$ and $U_{V + v}$.
    The blue area indicates a new job $\iota$ with volume $v$ and resource requirement $r$ that is scheduled via $\WFstep/(U_V, \iota, H)$.
    Depending on the resource requirement $r$, the yellow line enters the blue area exactly once, either on the upper plateau (\subref{fig:gt_incoming_job:a}) or on the lower plateau (\subref{fig:gt_incoming_job:b}).
}\label{fig:gt_incoming_job}
\end{figure}

We illustrate the possible situations in \cref{fig:gt_incoming_job}.
From the definition of \WFstep/, we know that there is exactly one time point $t^*$ at which the function $U - W$ switches signs (it goes from positive to negative).
With the notation $\intoo{x}_+ \coloneqq \max\set{x, 0}$, we have
\begin{dmath}
\Delta(C, y)
=
\int_0^C \intoo[\big]{U(t) - y}_+ - \intoo[\big]{W(t) - y}_+ \dt
\end{dmath}.
Using the monotonicity of both $W$ and $U$, we can consider their (say left-continuous) inverse functions $W^{-1}$ and $U^{-1}$, which allow us to compute the partial derivatives of $\Delta$ as
\begin{align}
\frac{\partial}{\partial C}\Delta(C, y)
&=
\intoo[\big]{U(C) - y}_+ - \intoo[\big]{W(C) - y}_+
\\\intertext{and}
\frac{\partial}{\partial \phantomas{C}{y}}\Delta(C, y)
&=
\min\set{W^{-1}(y), C} - \min\set{U^{-1}(y), C}
.
\end{align}
We use these in the remainder of the proof to gradually prove $\Delta(C, y) \geq 0$ for all $C \in \R_{\geq0}$ and $y \in \intcc{0, 1}$.

First, for any $y \in \intcc{0, 1}$ we obviously have $\Delta(0, y)  = A^0_U(y) - A^0_W(y) = 0 - 0 = 0$ by definition of $A^0_U(y)$ and $A^0_W(y)$.
Moreover, by definition of $t^*$, $U(C) - W(C) \geq 0$ for all $C \in \intcc{0, t^*}$, we have
\begin{math}
\frac{\partial}{\partial C}\Delta(C, y)
=
\intoo[\big]{U(C) - y}_+ - \intoo[\big]{W(C) - y}_+
\geq
\intoo[\big]{W(C) - y}_+ - \intoo[\big]{W(C) - y}_+
=
0
\end{math}.
So $\Delta(0, y) = 0$ and $C \mapsto \Delta(C, y)$ is non-decreasing on $\intcc{0, t^*}$, implying that $\Delta(C, y) \geq \Delta(0, y) \geq 0$ for all $C \in \intcc{0, t^*}$ and $y \in \intcc{0, 1}$.

Now fix any $C > t^*$ (such that $U(C) - W(C) \leq 0$ by definition of $t^*$).
For any $y \in \intcc[\big]{W(C), 1}$, we have $y \geq W(C) \geq U(C)$, yielding
\begin{math}
\frac{\partial}{\partial C}\Delta(C, y)
=
0 - 0
=
0
\end{math}.
So $\Delta(t^*, y) \geq 0$ and the function $C \mapsto \Delta(C, y)$ is non-decreasing on $\intcc{t^*, \infty}$, implying $\Delta(C, y) \geq \Delta(t^*, y) \geq 0$ for all $C \in \intcc{t^*, \infty}$ and $y \in \intcc[\big]{W(C), 1}$.

It remains to consider $C \in \intcc{t^*, \infty}$ and $y \in \intco[\big]{0, W(C)}$.
For the sake of a contradiction, assume there are $\bar{C} \in \intcc{t^*, \infty}$ and $\bar{y} \in \intco[\big]{0, W(\bar{C})}$ with $\Delta(\bar{C}, \bar{y}) < 0$.
Using that $W^{-1}$ is non-increasing, for any $y \in \intcc{0, \bar{y}}$ we get (since $y \leq \bar{y} \leq W(\bar{C})$) that $W^{-1}(y) \geq W^{-1}(W(\bar{C})) \geq \bar{C}$ (note that $W(\bar{C})$ might be at a discontinuity of $W^{-1}$).
Then we have
\begin{math}
\frac{\partial}{\partial \phantomas{C}{y}}\Delta(\bar{C}, y)
=
\min\set{W^{-1}(y), \bar{C}} - \min\set{U^{-1}(y), \bar{C}}
=
\bar{C} - \min\set{U^{-1}(y), \bar{C}}
\geq
0
\end{math}.
So $\Delta(\bar{C}, \bar{y}) < 0$ and the function $y \mapsto \Delta(\bar{C}, y)$ is non-decreasing on $\intcc{0, \bar{y}}$, implying $\Delta(\bar{C}, y) \leq \Delta(\bar{C}, \bar{y}) < 0$ for all $y \in \intcc{0, \bar{y}}$.
In particular, $\Delta(\bar{C}, 0) < 0$.
But then
\begin{math}
\frac{\partial}{\partial C}\Delta(C, 0)
=
\intoo[\big]{U(C) - y}_+ - \intoo[\big]{W(C) - y}_+
=
U(C) - W(C)
\leq
0
\end{math}
for all $t \in \intcc{\bar{C}, \infty}$.
So $\Delta(\bar{C}, 0) < 0$ and the function $C \mapsto \Delta(C, 0)$ is non-increasing on $\intcc{\bar{C}, \infty}$, implying $\Delta(\infty, 0) \leq \Delta(\bar{C}, 0) < 0$.
This clearly contradicts $\Delta(\infty, 0) = A^{\infty}_{U}(0) - A^{\infty}_W(0) = (V + v) - (V + v) = 0$, finishing the proof.

\section{Details for \Cref{sec:mintotalcompletiontime}}
\label{app:details_for_section_3}

\subsection{Proof of \Cref{prop:lower_bounds}}
\label{app:lowerbound12}

In this sub-section we prove \Cref{prop:lower_bounds}.
Lower Bounds (1) and (2) can be gleaned from
\textcite[Definitions 6 and 7]{beaumont2012minimizing}.
For sake of completeness, we give their proofs here.

\begin{proof}[Proof of Bounds (1) and (2) of \Cref{prop:lower_bounds}]
    For Bound (1),
    suppose that there is infinite resource available (instead of $1$).
    Then an optimal schedule $R^*$ will set $R^*_j(t)=\mathds{1}_{t\in \intco{0,p_j}}\cdot r_j$, so $C_j(R^*_j)=p_j$.
    Since all previously feasible schedules remain feasible,
    the optimal total completion time may only decrease.
    Hence, the total completion time of jobs $C^L=\sum_{j\in J} p_j$ is a lower bound on $C^*$, i.e., $C^*\ge C^L$.

    For Bound (2),
    suppose instead that we increase all resource requirements to $1$.
    Again, the optimal total completion time may at most decrease.
    The problem is now essentially $1|pmtn|\sum_{j\in J} C_j$, for which it is well-known that there is an optimal schedule that uses the SPT rule (shortest processing time first) \cite{Baker:74:Introduction-to-Sequencing}.
    This is equivalent to scheduling the jobs according to their volumes $v_1\le \dots \le v_n$ in ascending order.
    The total completion time then becomes $C^A=\sum_{j=1}^n \sum_{i=1}^{j} v_j$.
    This establishes our second lower bound: $C^*\ge C^A$.
\end{proof}

For Bound (3), we need a result by Sadykov \cite{sadykov2009scheduling}.
It states that there is a schedule with minimum total completion time that has the
\emph{ascending property}, i.e., each resource assignment $R_j$ is non-decreasing until the respective job $j$ completes.
\begin{lemma}[{\cite[Theorem~1]{sadykov2009scheduling}}]%
    \label{lem:sadykov}
    There exist an optimal schedule $R^*$ for total completion time minimization such that $R^*_j(t)\le R^*_j(t')$ for all $t\le t'\le C_j(R^*)$ and $j\in J$.\footnote{Their model limits the number of machines $m$. We can effectively assume $m=\abs{J}$.}
\end{lemma}

We then show that at least half of the volume
of each job lies after the job's fractional completion time (see \Cref{lem:volume_distribution_for_ascending_schedules}).
The proof can be found below.
\begin{restatable}{lemma}{lemvolumedistributionforascendingschedules}
    \label{lem:volume_distribution_for_ascending_schedules}
    In ascending schedules $R$, jobs $j\in J$ schedule $\ge v_j/2$ volume after $C^F_j(R)$.
\end{restatable}

% \peter[inlinepar]{
%     We need some text here or some other way to avoid that the following proof looks like it is for the above lemma.
% }

Using this lemma, we can give the proof of the Bound (3) of \Cref{prop:lower_bounds}.
\begin{proof}[Proof of Lower Bound (3) of \Cref{prop:lower_bounds}]
    By \Cref{lem:sadykov}, there exists an optimal schedule $R^*$ that has the ascending property. By \Cref{lem:volume_distribution_for_ascending_schedules}, $R^*$ schedules at least $v_j/2$ volume of each $j\in J$ after $C_j^F(R^*)$.
    So $j$ requires at least $p_j/2$ units of time after $C_j^F(R^*)$ to finish since $R_j(t)\le r_j$ for all $t\ge 0$.
    Hence $C_j(R^*)\ge C_j^F(R^*)+p_j/2$.
    Summing over all $j\in J$ yields
    $C^*=C(R^*)\ge C^F(R^*)+1/2\cdot \sum_{j\in J}p_j=C^{F*}+1/2\cdot C^L$.
\end{proof}

\begin{proof}[Proof of \Cref{lem:volume_distribution_for_ascending_schedules}]
    Let $V_j(R)$ be the volume scheduled for $j$ before $C_j^F(R)$.
    The statement is equivalent to showing that $V_j(R)\le v_j/2$.
    For that we construct a new resource assignment $\tilde{R}_j$ with $C_j^F(\tilde{R})\geq C_j^F(R)$.
    It is constructed from $R$ by rescheduling the volume around $C_j^F(R)$ such that always $r\coloneqq R_j(C_j^F(R))\ne 0$ resource is used.
    Formally, we set $\tilde{R}_j(t)=r\cdot \mathds{1}_{a_j\le t<b_j}$,
    where $a_j:=C_j^F(R)-V_j(R)/r$ and $b_j:=C_j^F(R)+(v_j-V_j(R))/r$. This means that the volume before (after) $C_j^F(R)$ stays before (after) $C_j^F(R)$, respectively.
    
    Because $R$ is ascending, we only shift volume from earlier to later time points in $\tilde{R}$ compared to $R$.
    From this $C_j^F(\tilde{R})\geq C_j^F(R)$ follows directly.
    We calculate $C_j^F(\tilde{R})$.
    It is
    \begin{align*}
        C_j^F(\tilde{R})
        &=\int_{0}^{\infty}{\frac{t\cdot \tilde{R}_j(t)}{v_j}\dt}
        =\frac{r}{v_j}\int_{a_j}^{b_j}{t\dt}
        =\frac{r}{v_j}\cdot \frac12(b_j^2-a_j^2)
        =\frac{r}{2v_j}\cdot (b_j-a_j)(a_j+b_j)\\
        &=\frac{r}{2v_j}\cdot \frac{v_j}{r} \left(
        C_j^F(R)-\frac{V_j(R)}{r}+C_j^F(R)+\frac{v_j-V_j(R)}{r}
        \right)
        =C_j^F(R)+\frac{v_j-2V_j(R)}{2r}
    \end{align*}
    Combining this with $C_j^F(\tilde{R})\geq C_j^F(R)$ yields $V_j(R)\leq v_j/2$, which is the desired result.
\end{proof}

\subsection{Analysis of \textsc{Greedy}}
\label{app:greedy}

In this section we show \Cref{prop:greedy_bound}.
Throughout this sub-section, we assume that all $j\in J$ are ordered as algorithm \textsc{Greedy} sorts them, i.e.,
$v_1\le \dots\le v_n$.
Imagine we cut \textsc{Greedy}'s schedule $R^G$ at specific time points $0=\tau_0<\dots<\tau_m=M(R^G)$.
We then observe for each $i\in [m]$ the sub-schedule $R^{\tau_i}$ that contains all job volumes scheduled in the time interval $[0,\tau_i)$ in $R^G$, respectively.
We can then associate each sub-schedule with its total completion time $C(R^{\tau_i})$ by
only looking at the portions of jobs scheduled and ignoring all so-far unscheduled jobs.
At the same time, we consider the lower-bound equivalents from \Cref{prop:lower_bounds} for these job portions, i.e., $C^L(\tau_i)$ and $C^A(\tau_i)$ (see below for formal notations).
We can then easily see that $C^L(\tau_0)=C^A(\tau_0)=C(R^{\tau_0})=0$ as well as $C^L(\tau_m)=C^L$, $C^A(\tau_m)=C^A$ and $C(R^{\tau_m})=C(R^G)$.
By inductive application of the following \Cref{lem:iteration}, \Cref{prop:greedy_bound} follows.

\begin{lemma}
    \label{lem:iteration}
    If $C(R^{\tau_i})\le C^L(\tau_i)+C^A(\tau_i)$ for $i<m$, then also $C(R^{\tau_{i+1}})\le C^L(\tau_{i+1})+C^A(\tau_{i+1})$.
\end{lemma}

We prove \Cref{lem:iteration} after giving some additional notation and observations about \textsc{Greedy}'s schedules.

\paragraph{Additional Notation}
\label{def:greedy_subschedules}
Denote $0=\tau_0<\dots<\tau_m=M(R^G)$ where $\tau_i$, $i\in [m]$, denotes the $i$'th smallest distinct completion time in $R^G$.
Let $R^\tau(t)\coloneqq R^G(t) \cdot \mathds{1}_{t<\tau}$ be the (sub-)schedule of $R^G$ up to time point $\tau\in\mathbb{R}_{\ge 0}$.
It schedules exactly $v_j(\tau)\coloneqq\int_{0}^{\tau}{R^G_j(t)\dt}$ volume
for each job $j\in J$.
We use $J^\tau\coloneqq\set{j\in J | v_j(\tau)>0}$ to denote the set of $n(\tau)\coloneqq\abs{J^{\tau}}$ jobs scheduled by $R^\tau$.
From \Cref{prop:lower_bounds}, we get the lower bounds
$C^L(\tau)\coloneqq\sum_{j\in J}{v_j(\tau)/r_j}$
and
$C^A(\tau)\coloneqq\sum_{j=1}^{n(\tau)}\sum_{i=1}^j v_j(\tau)$
for the optimal total completion time of $J^\tau$.

\begin{restatable}{observation}{obsgreedyproperties}
    \label{obs:greedy_properties}
    The solution $R^G$ produced by \textsc{Greedy} has the following properties:
    \begin{enumerate}
        \item {
            $R^G$ stays constant within each time interval $[\tau_{i-1},\tau_i)$ for all $i\in [m]$.
        }
        \item {
            At any time point $t$, $R^G$ has at most one job $j$ that does not receive its full resource requirement, i.e., $R^G_j(t)<r_j$.
            Furthermore, $j$ has the highest $v_j$ among all jobs scheduled somewhere within $[0,t)$.
        }
    \end{enumerate}
\end{restatable}

\Cref{obs:greedy_properties} is straightforward to prove.
Using this \lcnamecref{obs:greedy_properties}, we can now prove \Cref{lem:iteration}.
\begin{proof}[Proof of \Cref{lem:iteration}]
    We determine how $C(R^\tau)$, $C^L(\tau)$ and $C^A(\tau)$ change as we increase $\tau$ from $\tau_i$ to $\tau_{i+1}$.
    We can essentially view this as a two-step process.
    First, some new jobs may be started at $\tau_i$, say jobs $J_i^+\coloneqq \set{k,k+1,\dots,l-1}$ ($J_i^+=\varnothing$ possible).
    This changes $C(R^\tau)$, as we know have to add the completion times of the jobs in $J_i^+$.
    At the same time, $n(\tau)$ and $C^A(\tau)$ change. $C^L(\tau)$ does not change.
    Second, we increase the volume scheduled for the jobs scheduled within $[\tau_i,\tau_{i+1})$ (which we denote by $J_i$) by increasing $\tau$ to $\tau_{i+1}$.

    Assume first that $\bar{R}^G(t)=1$ for $t\in \intco{\tau_i,\tau_{i+1}}$.
    For the first step, when we add the jobs $J_i^+$, $C(R^\tau)$ increases by $\abs{J_i^+}\cdot \tau_i=(l-k)\cdot \tau_i$.
    Similarly,
    $C^A(\tau)$ increases by $(l-k)\cdot \sum_{j=k}^{l-1} \sum_{i=1}^j v_j(\tau_i)=(l-k)\cdot \sum_{i=1}^j v_j(\tau_i)=(l-k)\cdot \tau_i$, where the last equality comes from the fact that $\bar{R}^G(t)=1$ for all $t<\tau_i$.
    This establishes that $C(R^\tau)$ and $C^A(\tau)$ change by the same amount (and $C^L(\tau)$ does not change).

    For the second step, recall \Cref{obs:greedy_properties}.
    Because $R^G$ stays constant within $\intco{\tau_{i},\tau_{i+1}}$,
    $C(R^\tau)$ increases by $\abs{J_i}\cdot (\tau_{i+1}-\tau_i)$, where $p$ is the number of jobs scheduled within this interval.
    Statement (2) assures that all jobs in $J_i$ receive their full resource requirement within $\intco{\tau_{i},\tau_{i+1}}$, except for the one with the highest index.
    Therefore, $C^L(\tau)$ increases by at least $(\abs{J_i}-1)\cdot (\tau_{i+1}-\tau_i)$:
    This is because each job in $J_i$ that receives full resource requirement in $\intco{\tau_{i},\tau_{i+1}}$ will increase
    $C^L(\tau)$ by $(v_j(\tau_{i+1})-v_j(\tau_i))/r_j=(\tau_{i+1}-\tau_i)\cdot r_j/r_j=\tau_{i+1}-\tau_i$.
    For $C^A(\tau)$, we see an increase by at least $\sum_{i=1}^{l-1} v_j(\tau_{i+1})-\sum_{i=1}^{l-1} v_j(\tau_{i})=\tau_{i+1}-\tau_i$. To summarize, $C(R^\tau)$ increases at most as fast as $C^L(\tau)+C^A(\tau)$ when increasing $\tau$ from $\tau_i$ to $\tau_{i+1}$.

    Instead assume that $C(\bar{R}(t))\ne 1$ for $t\in \intco{\tau_i,\tau_{i+1}}$.
    Then all jobs will receive their resource requirement by definition of \textsc{Greedy},
    and $C^L(\tau)$ increases at least as much as $C(R^\tau)$, analogous to above argument.
\end{proof}

\subsection{Generalized CLP and DCP}
\label{subsec:generalized_clp_dcp}

For the purpose of proving \Cref{lem:exact_volume_completion} (see \Cref{subsec:lem_exact_volume_completion}), we first extend the definition of $CLP$ and $DCP$ from
\Cref{sec:mintotalcompletiontime}.
In these generalizations, we change the volume that is required for scheduling from a vector $v$ of volumes to a vector of volumes $\bar{v}$. However, the volumes within the $CLP$ objective remains intact. The $CLP$ then becomes

\label{def:CLP_general}
\begin{align*}
    \text{minimize } &\sum_{j\in J}\int_0^\infty{\frac{t\cdot R_j(t)}{v_j}\dt}~~~~~~~~&
    \int_0^\infty{R_j(t)\dt}\ge \bar{v}_j~~\forall j\in J\\
    &0\le R_j(t)\le r_j~~\forall j\in J, t\in \mathbb{R}_{\ge 0}~~~~~~~~&
    \sum\nolimits_{j\in J} R_{j}(t)\le 1~~\forall t\in \mathbb{R}_{\ge 0}
    %\text{\footnoteref{ftn:regularity}}
\end{align*}
and the $DCP$ becomes
\label{def:DCP_general}
\begin{align*}
    \text{maximize } &\sum_{j\in J} \alpha_j \bar{v}_j - \sum_{j\in J}r_j \int_{0}^{\infty} \beta_{j}(t) \dt-\int_{0}^{\infty}\gamma(t)\dt\\
    \text{s.t. }&\alpha_j,\beta_j(t),\gamma(t)\ge 0~~\forall j\in J,t\in \mathbb{R}_{\ge 0}~~~~~~
    \gamma(t)+\beta_j(t)\ge \alpha_j-t/v_j~~\forall j\in J,t\in \mathbb{R}_{\ge 0}.
\end{align*}
We denote these two systems by $CLP(\bar{v})$ and $DCP(\bar{v})$.

\paragraph{Time-Discretized LP and its dual}
\label{par:time_discretized_lp_and_its_dual}

We obtained $DCP(\bar{v})$ by first looking at the time-discretized version of $CLP(\bar{v})$.
In this discretization, we subdivide time into equal-sized \emph{slots} inside of which we assume the resource assignments to be finite.
To get a linear program, we need to fix a time horizon $T$ where no job is scheduled past $T$ in the optimal continuous solution.
This can be guaranteed by setting $T$ to be at least $n\cdot p_{\mathrm{max}}$, where $p_{\mathrm{max}}=\max\set{p_j|j\in J}$.
For some slot width $\delta>0$ (that divides $T$), we then let
$I\coloneqq [T/\delta]$ be the set of \emph{slots}.
Similar approaches are known from, e.g., \cite{buie1973numerical, wen2010recurrence, wen2011using, wen2012approximate}.

As variables, we will not use the resource assignments of jobs in slots, but instead the volume $V_{j,i}$ that jobs
$j$ schedule in slots $i$. This way, it will be easier to translate it to the continuous versions.
For dualization, we provide the Lagrange dual variables
$\alpha_j,\beta_{j,i},\gamma_i$ ($j\in J,i\in I$) corresponding to the constraints.

The $LP$ is then
\begin{align*}
    \text{minimize } &\sum_{j\in J}\frac{1}{v_j}\sum_{i\in I}{V_{j,i}\cdot \left(i\delta-\frac\delta2\right)}~~~~~~~~~~~~~~~~~~~~~~
    \sum_{i\in I}{V_{j,i}}\ge \bar{v}_j~~\forall j\in J \rightsquigarrow\alpha_j\\
    &0\le V_{j,i}\le r_j\cdot \delta~~\forall j\in J, i\in I\rightsquigarrow\beta_{j,i}~~~~~\,~~~
    \sum_{j\in J} V_{j,i}\le \delta~~\,~\forall i\in I\rightsquigarrow\gamma_i
\end{align*}

and the dual $LP$ is

\begin{align*}
    \text{maximize } &\sum_{j\in J} \alpha_j \bar{v}_j - \sum_{j\in J}r_j \sum_{i\in I} \beta_{j,i} \cdot \delta-\sum_{i\in I}\gamma_i\cdot \delta\\
    &\alpha_j,\beta_{j,i},\gamma_i\ge 0~~\forall j\in J,i\in I\\
    &\gamma_{i}+\beta_{j,i}\ge \alpha_j-(i\delta-\delta/2)/v_j~~\forall j\in J,i\in I
\end{align*}

Notice that the terms $(i\delta-\delta/2)$, $i\in I$, represent the midpoints of the $i$'th slot, respectively.
In the primal objective, it is used for a rectangular integration using samples at the slot midpoints.
The corresponding dual constraint is the same as in the continuous version, also restricted to the slot midpoints.
When $\delta$ approaches zero, the rectangular integration will effectively become an integral, and
the constraints/slackness conditions will hold for all $t\ge 0$.
For the constraints $0\le V_{j,i}\le r_j\cdot \delta$ and $\sum_{j\in J} V_{j,i}\le \delta$, it is easy to
see that a division by $\delta$ effectively yields its continuous counterpart.

\paragraph{Slackness Conditions and Primal-Dual-Pairs}
\label{par:slackness_and_primaldual}
The LP slackness condition are as follows.
\begin{enumerate}
    \item {
        $a_j(\bar{v}_j-\sum_{i\in I}{V_{j,i}})=0~~\forall j\in J$
    }
    \item {
        $\beta_{j,i}(r_j\cdot \delta-V_{j,i})=0~~\forall j\in J,i\in I$
    }
    \item {
        $\gamma_i(\delta-\sum_{j\in J} V_{j,i})=0~~\forall i\in I$
    }
    \item {
        $V_{j,i}(\alpha_j-(i\delta-\delta/2)/{v_j}-\beta_{j,i}-\gamma_i)=0~~\forall j\in J,i\in I$
    }
\end{enumerate}

For $CLP(\bar{v})$/$DCP(\bar{v})$, we establish the following continuous slackness conditions. For $\bar{v}=v$, these become equivalent to the slackness conditions in \Cref{subsec:fractional}.

\begin{enumerate}
    \item {
        ($\alpha$-slackness condition): $\alpha_j(\bar{v}_j-\int_{0}^{\infty}{R_j(t)\dt})=0~~\forall j\in J$
    }
    \item {
        ($\beta$-slackness condition): $\beta_j(t)(r_j-R_j(t))=0~~\forall j\in J,t\in \mathbb{R}_{\ge 0}$
    }
    \item {
        ($\gamma$-slackness condition): $\gamma(t)(1-\sum_{j\in J} R_j(t))=0~~\forall t\in \mathbb{R}_{\ge 0}$
    }
    \item {
        ($R$-slackness condition): $R_j(t)(\alpha_j-t/v_j-\beta_j(t)-\gamma(t))=0~~\forall j\in J,t\in \mathbb{R}_{\ge 0}$
    }
\end{enumerate}

A \emph{primal-dual-pair} for $CLP(\bar{v})$/$DCP(\bar{v})$ is a 5-tuple $(R,\alpha,\beta,\gamma,\bar{v})$ consisting of a schedule $R$, a dual schedule $\alpha$, a tuple of functions $(\beta_j)_{j\in J}$ ($\beta_j: \mathbb{R}_{\ge 0}\to \mathbb{R}_{\ge 0}$ and a function $\gamma: \mathbb{R}_{\ge 0}\to \mathbb{R}_{\ge 0}$, such that above slackness conditions are fulfilled (for volumes $\bar{v}=(\bar{v}_j)_{j\in J}$).

\subsection{Proof of \Cref{lem:exact_volume_completion}}
\label{subsec:lem_exact_volume_completion}

The purpose of this sub-section is to prove the following
\lcnamecref{lem:exact_volume_completion}.

\lemexactvolumecompletion*

We require a slightly more general version of this lemma in \Cref{sec:threehalfspluseps}.
This generalization is based on the continuous linear programs $CLP(\bar{v})$ and $DCP(\bar{v})$ defined in \Cref{subsec:generalized_clp_dcp}.
With this, the statement we want to show for this sub-section now becomes

\begin{restatable}{lemma}{lemexactvolumecompletiongeneral}
    \label{lem:exact_volume_completion_general}
    For any job set $J$ (with volumes $v$) and for any volumes $\bar{v}$ there exists an $\alpha$ such that the line schedule of $\alpha$ is a primal-dual pair for $CLP(\bar{v})$ and $DCP(\bar{v})$.
\end{restatable}

Note that the concept of a \emph{line schedule} remains the same as in \Cref{sec:mintotalcompletiontime}, however for a primal-dual-pair, we require that $\int_{0}^\infty R_j(t)\dt=\bar{v}_j$ for all $j\in J$.

Our proof idea for \Cref{lem:exact_volume_completion_general} is as follows.
Consider an arbitrary dual schedule $\alpha$.
We can use \Cref{def:line_schedule} to extend $\alpha$ to a line schedule $(R,\alpha,\beta,\gamma,\bar{v})$.
This allows us to define a function, mapping from $\alpha$ to the scheduled volumes as $v(\alpha)\coloneqq \int_{0}^\infty R_j(t)\dt$.
Our general idea is then to show that it is possible to guide $\alpha$ such that $v(\alpha)$ converges to the desired volumes $\bar{v}$, which then proves \Cref{lem:exact_volume_completion_general}.

For its proof, we require the two following lemmas. In \Cref{lem:dual_schedule_completion}, we show that
any line schedule as constructed in \Cref{def:line_schedule} is a primal-dual-pair. In \Cref{obs:continuous_changing_v_bar}, we show various properties about the function $v(\cdot)$ defined above.

\begin{lemma}
    \label{lem:dual_schedule_completion}
    The line schedule of any vector $(\alpha_j)_{j\in J}$
    is a primal-dual pair.
\end{lemma}

\begin{lemma}
    \label{obs:continuous_changing_v_bar}
    The following observations hold for the function $v(\cdot)$ and a vector $\alpha$:
    \begin{itemize}
        \item {
            $v(\cdot)$ is continuous.
        }
        \item {
            If $\alpha_j=0$ for some $j\in J$, then $v_j(\alpha)=0$.
        }
        \item {
            If $v_j(\alpha)\le x$ for some $j\in J$ and $x\ge 0$, then there exists $\hat{\alpha}_j\ge 0$ such that replacing $\alpha_j$ with $\hat{\alpha}_j$ in $\alpha$ gives a new vector $\tilde{\alpha}$ with $v_j(\tilde{\alpha})=x$.
        }
        \item {
            If we increase increase $\alpha_j$ to obtain $\tilde{\alpha}$, then
            $v_j(\tilde{\alpha})\ge v_j(\alpha)$ and $v_{j'}(\tilde{\alpha})\le v_{j'}(\alpha)$ for all $j'\in J\setminus\set{j}$.
        }
    \end{itemize}
\end{lemma}

Equipped with these two lemmas, we can now provide the proof of \Cref{lem:exact_volume_completion_general}.

\begin{proof}[Proof of \Cref{lem:exact_volume_completion_general}]
    Let $\alpha^{(0)}=\mathbf{0}$, i.e., all dual lines pass through the origin.
    We will (recursively) define $\alpha^{(i+1)}$ from $\alpha^{(i)}$ ($i\in\mathbb{N}$).
    For any $\alpha$ with $v(\alpha)\le \bar{v}$, define $\alpha^+$ where
    each $\alpha_j$, $j\in J$, is replaced by $\hat{\alpha}_j$ according to
    the third statement of \Cref{obs:continuous_changing_v_bar}, i.e., we raise each dual line
    such that the respective job would receive $\bar{v}_j$ volume if considered in isolation.
    Then define $\alpha^{(i+1)}={\alpha^{(i)}}^+$.

    Notice that the sequence $(\alpha^{(i)})_{i\ge 0}$ is monotonous non-decreasing.
    Specifically, if $v_j(\alpha^{(i)})\ne \bar{v}_j$, then $\alpha^{(i+1)}\ne \alpha^{(i)}$.
    Furthermore, because of the third statement in \Cref{obs:continuous_changing_v_bar},
    it follows from $v_j(\alpha^{(i)})\le \bar{v}_j$ that $v_j(\alpha^{(i+1)})\le \bar{v}_j$.
    From this also follows that all fix points $\alpha$ of $\alpha^+$ must satisfy $v(\alpha)=\bar{v}$.

    If we can show that the sequence is also bounded, then by the monotone convergence theorem, the sequence then converges to its supremum, which we call $\alpha^*=\lim_{i\rightarrow\infty} \alpha^{(i)}$.
    Since the sequence can only converge to a fix point, we then get $v_j(\alpha^*)=\bar{v}_j$.
    Using \Cref{lem:dual_schedule_completion}, the line schedule of $\alpha^*$ is then a primal-dual pair for $CLP(\bar{v})$ and $DCP(\bar{v})$.

    So it only remains to show that $(\alpha^{(i)})_{i\ge 0}$ is bounded.
    Specifically, we show that, for each $j\in J$, $\alpha_j\le (\sum_{j'\in J}{\bar{v}_{j'}})/(v_j \cdot \min_{j'\in J} r_{j'})$.
    Suppose the opposite.
    Then, by the definition of a line schedule,
    for each $t< \alpha_j v_j$, there is at
    least one job scheduled at time $t$ that receives at least $\min_{j'\in J} r_{j'}$ resource. Then, a total volume of at least $\alpha_j v_j \cdot \min_{j'\in J} r_{j'}>\sum_{j'\in J}{\bar{v}_{j'}}$ is scheduled, with a contradiction. With this, we established that there exists an upper bound on each $\alpha_j$ in the process, thus finishing the proof.
\end{proof}

It remains to give the proofs for the remaining two lemmas.
In \Cref{lem:dual_schedule_completion}, we essentially have to check the slackness conditions. For \Cref{obs:continuous_changing_v_bar}, the desired statements stem from geometic observations about line schedules.

\begin{proof}[Proof of \Cref{lem:dual_schedule_completion}]
    We show that the slackness conditions from \Cref{par:slackness_and_primaldual} are fulfilled.
    Let $(R,\alpha,\beta,\gamma,\bar{v})$ be the line schedule of $\alpha$.
    The $\alpha$-slackness conditions are trivially fulfilled by the definition of $\bar{v}$.
    For any $\beta$-slackness condition, let first $t\in \mathbb{R}_{\ge 0}$. Observe that if $\beta_j(t)>0$ (else its trivially true), then $\beta_j(t)=d_j(t)-\gamma(t)>0$ and $d_j(t)>0$.
    Because of the choice of $\gamma$, we must have that
    $j$ received its full resource requirement at $t$, i.e., $R_j(t)=r_j$, which fulfills the condition.
    For the $\gamma$-slackness condition,
    if $\gamma(t)=0$ then the condition is fulfilled, otherwise
    since $\gamma(t)=d_j(t)>0$ for some job $j$, and by definition
    of $R_j(t)$, the resource must have been exhausted, i.e., $\bar{R}_j(t)=1$, which means that the $\gamma$-slackness condition is true.
    Finally, for the $R$-slackness condition,
    when a job is scheduled at time $t$,
    then $\gamma(t)\le d_j(t)$ and therefore $\beta_j(t)= d_j(t)-\gamma(t)$.
    This implies $\alpha_j-t/v_j-\beta_j(t)-\gamma(t)=0$ which fulfills such conditions as well.
\end{proof}

\begin{proof}[Proof of \Cref{obs:continuous_changing_v_bar}]
    1.
    If $\alpha$ is changed continuously, then all dual lines change continuously.
    Since the job set is non-degenerate, all dual lines have different slope, so their intersections change continuously with $\alpha$. Then also the time points where $R_j$ changes over time changes continuously, and as such also $\bar{v}$.
    
    2.
    This statement holds trivially by definition of a line schedule.
    
    3.
    We can set $\alpha_j$ such that $d_j(t)>d_{j'}(t)$ for all $t\in \intco{0,\max\set{\alpha_j v_j|j\in J}}$ and $j'\ne j$, i.e., the dual line $d_j(t)$ lies above all other dual lines until all dual lines fall below zero.
    As such, by raising $\alpha_j$, we can make $\bar{v}_j(\alpha')$ arbitrarily large.
    Because $\bar{v}$ is continuous over $\alpha$, the intermediate value theorem guarantees the existence of an $\alpha_j$ such that the corresponding $\alpha'$ has $\bar{v}_j(\alpha')=v_j$.
    
    4.
    As we increase $\alpha_j$, the value of the dual line $d_j(t)$ for each time point $t$ is only increasing.
    As such $j$ increases as of the order $\succ_t$ (defined in \Cref{def:line_schedule}) and as such only increases its resource assignment for all $t\ge 0$.
    The converse is true for all other jobs $j'\ne j$.
    They may at most decrease according to $\succ_t$ and thus reduce their $R_j(t)$.
    From this, the statement immediately follows by the definition of $v(\cdot)$.
\end{proof}

\subsection{Calculation for \Cref{fig:dual_line_schedule}}
\label{calc:fig:dual_line_schedule}
To show that each job schedules its volume, we first calculate the intersections between the dual lines.
Two dual lines for jobs $j,j'$ intersect when
$\alpha_j-t/v_j=\alpha_{j'}-t/v_{j'}$, i.e., when
$t=(\alpha_j-\alpha_{j'})/(1/v_j-1/v_{j'})$.
We get the intersections $t_{j,j'}$ between jobs $j,j'$ as
$t_{1,2}=(51/16-39/16)/(1/1-1/4)=1$, $t_{1,3}=(51/16-31/16)/(1/1-1/6)=3/2$ and $t_{2,3}=(39/16-31/16)/(1/4-1/6)=6$.
$j_1$ receives $r_1=3/4$ resource within $\intco{0,t_{1,2}}=\intco{0,1}$ and $1-r_2=1/2$ resource within $\intco{t_{1,2},t_{1,3}}=\intco{1,3/2}$.
The volume scheduled is $3/4\cdot (1-0)+1/2\cdot(3/2-1)=1=v_1$.
Similarly, for $j_2$, $(1-r_1) (t_{1,2}-0)+r_2\cdot (t_{2,3}-t_{1,2})+(1-r_3)\cdot (\alpha_2 v_2-t_{2,3})=1/4\cdot 1+1/2 \cdot 5+1/3\cdot 15/4=4=v_2$;
and for $j_3$, $(1-r_2)\cdot (t_{2,3}-t_{1,3})+r_3\cdot (\alpha_3 v_3 - t_{2,3})=1/2 \cdot 9/2+2/3\cdot 45/8=6=v_3$.

\subsection{Geometric Shapes and their properties}
\label{subsec:geometric_shapes}

This sub-section aims to show the three lemmas about the 3D shapes suggested in \Cref{subsec:fractional}, specifically \Cref{lem:no_overlap, lem:alpha_shapes_composed_of_other_shapes, lem:shape_equalities}.
Before we can show these, we need to become more concrete about the 3D shapes.
For the purpose of \Cref{sec:threehalfspluseps}, we use the more general $CLP(\bar{v})$/$DCP(\bar{v})$ definitions.
In particular, we define the objective quantities $\Rho=\sum_{j\in J}\int_0^\infty t\cdot R_j(t)/v_j\dt$ (the primal objective) and
$\Alpha=\sum_{j\in J} \alpha_j \bar{v}_j$, $\Beta=\sum_{j\in J}r_j\int_0^\infty\beta_j(t)\dt$ and $\Gamma=\int_0^\infty\gamma(t)\dt$ (the parts of the dual objective).

For the 3D shapes, we will first define a geometrical representation of a (primal) schedule (see \Cref{def:omega_areas}).
Each job is assigned a subset of points in $\intco{0,\infty}\times \intco{0,1}$ that indicates where that job is scheduled.
We want to construct primal and dual volume pieces for jobs $j$ such that the time points $t$ where they start/end satisfy $d_j(t)=\gamma(t)$. This is justified by statement 2 of the following \Cref{obs:completion_properties}.
That is why we want each job to be assigned the same portion of the resource axis $\intco{0,1}$ at least until $d_j(t)=\gamma(t)$.

\begin{observation}
    \label{obs:completion_properties}
    A line schedule
    $(R,\alpha,\beta,\gamma,\bar{v})$ for a job set $J$ has the following properties.
    \begin{enumerate}
        \item {
            $\gamma$ is continuous, and $\gamma$ is strictly monotonically decreasing in an interval $\intco{0,t_\gamma}$ for some $t_{\gamma}\ge 0$, and $\gamma(t)=0$ for all $t\ge t_{\gamma}$.
        }
        \item {
            If for each $\varepsilon>0$ there is some $t'\in \intco{t-\varepsilon,t}$
            with $R_j(t')\ne R_j(t)$, i.e., $j$'s resource assignment changes at time $t$, then $d_j(t)=\gamma(t)$.
        }
    \end{enumerate}
\end{observation}

The proof can be found below.
Formally, we define the point subsets as follows.

\begin{definition}
    \label{def:omega_areas}
    Let $(R,\alpha,\beta,\gamma,\bar{v})$ be a completion, and let
    $0=t_0<t_1<\dots<t_m$ be the time points where $R$ changes, i.e.,
    such that $R$ is constant in each interval $\intco{t_{i-1},t_i}$ and $R(t_{i-1})\ne R(t_i)$ for all $i\in [m]$.
    A \emph{geometrical representation} $\Omega=(\Omega_j)_{j\in J}$ of $R$ consists of point sets $\Omega_j\subseteq \intco{0,\infty}\times \intco{0,1}$ for each $j\in J$, such that the following properties hold.
    \begin{enumerate}
        \item {
            $\set{r|(t,r)\in \Omega_j}$ is measureable and its measure is $R_j(t)$ for all $t\ge 0$.
            %\florian{volume? oder size? oder measure?}%Christoph: habe jetzt measure genommen.
        }
        \item {
            For all $t\ge 0$ and $j\in J$,
            if for each $\varepsilon>0$ there exists a $t'\in\intco{t-\varepsilon,t}$ such that
            $\set{r|(t,r)\in \Omega_j}\ne \set{r|(t',r)\in \Omega_j}$,
            then $\gamma(t)=d_j(t)$.
        }
    \end{enumerate}
\end{definition}

Such a geometrical representation can always be constructed.
When jobs gain or lose resource at some time point,
we make sure that only their portion of the resource axis $\intco{0,1}$ is involved at that time point.
As such the second condition of a geometrical representation can always be fulfilled.

With this, we are now ready to define the individual 3D shapes.

\begin{definition}
    \label{def:shapes}
    Let $\Omega$ be a geometrical representation of a completion $(R,\alpha,\beta,\gamma,\bar{v})$.
    For $0=r^0<r^1<\dots<r^k=1$, define strips
    $\Omega^{i}$ ($i\in [k]$) for some $k\in\mathbb{N}$
    where $\Omega^{i}=(\Omega^{i}_j)_{j\in J}$ and $\Omega^{i}_j=\set{(t,r)\in\Omega_j|r\in [r^{i-1},r^{i})}$, such that for each $t\in\mathbb{R}_{\ge 0}$,
    $\abs{\set{j\in J|(t,.)\in \Omega^{i}_j}}\le 1$.
    Define $R^{i}=(\abs{\Omega^{i}_j(t)})_{j\in J}$ where $\Omega^{i}_j(t')\coloneqq \set{(t,r)\in\Omega^{i}_j|t=t'}$. Then we require that if
    $R^{i}_j(t)>0$ for some $t\in\mathbb{R}_{\ge 0}$, then $R^{i}_j(t)=q^i\coloneqq r^{i}-r^{i-1}$.\footnote{We assume that the number of such strips is finite, and that the involved sets $\Omega_j^i(t)$ are measureable, which can be ensured by choosing a proper geometrical representation.}

    Consider any strip $\Omega^{i}$.
    Subdivide it by time points
    $0=t^{i,0}<\dots<t^{i,K^{i}}$ such that each
    $T^{i,l}\coloneqq [t^{i,l-1},t^{i,l})$ ($l\in [K^{i}]$) is an inclusive-maximal time interval where some job is scheduled within $\Omega^{i}$.
    Let the job scheduled in the time interval $T^{i,l}$ be called $j^{i,l}$.
    We define for each $i\in [k],l\in [K^i]$ the following \emph{$\Psi$-pieces} ($\Psi\in\set{\Rho,\Alpha,\Beta,\Gamma}$)
    \begin{enumerate}
        \item {
            $\Rho^{i,l}=\set{(t,r,\alpha)|(t,r)\in \Omega_{j^{i,l}}^{i}, d_{j^{i,l}}(t) \le \alpha< \alpha_{j^{i,l}},t\in T^{i,l}}$
        }
        \item {
            $\Alpha^{i,l}=\set{(t,r,\alpha)|(t,r)\in \Omega_{j^{i,l}}^i,0\le \alpha< \alpha_{j^{i,l}},t\in T^{i,l}}$
        }
        \item {
            $\Beta^{i,l}=\set{(t,r,\alpha)|r\in\Omega_{j^{i,l}}^{i}(t^{i,l-1}),\gamma(t^{i,l})\le \alpha< \min(\gamma(t^{i,l-1}),d_{j^{i,l}}(t)),\alpha\ge \gamma(t),t\ge 0}$
        }
        \item {
            $\Gamma^{i,l}=\set{(t,r,\alpha)|r\in\Omega_{j^{i,l}}^{i}(t^{i,l-1}),\gamma(t^{i,l})\le \alpha< \min(\gamma(t^{i,l-1}),d_{j^{i,l}}(t)),\alpha< \gamma(t),t\ge 0}$
        }
    \end{enumerate}

    Based on this, we define suitable shapes per job. For all $\Psi\in \set{\Rho,\Alpha,\Beta,\Gamma}$ and $j\in J$, we define the $\Psi$-shape $\Psi_j=\bigcup_{i\in [k],l\in [K^i],j^{i,l}=j} \Psi^{i,l}$. Further define $\Psi^{\mathrm{all}}=\bigcup_{i\in [k],l\in [K^i]}\Psi^{i,l}$.
\end{definition}

\begin{proof}[Proof of \Cref{obs:completion_properties}]
    By \Cref{lem:dual_schedule_completion}, $(R,\alpha,\beta,\gamma,\bar{v})$ is also a primal-dual pair and as such the slackness conditions are fulfilled.
    1.
    Start at $t=0$ and observe how $\gamma$ changes as we increase $t$. $\gamma$ will always coincide with some dual line $d_j$ or its value will be zero by definition.
    As long as no dual line crosses $d_j$, dual lines above (below) $d_j$ stay above (below) $d_j$, respectively. Therefore the resource assignment of the respective jobs stays constant, and $\gamma$ will coincide with the same dual line.
    As such $\gamma$ is monotonically decreasing until it becomes zero.
    It remains to show that $\gamma$ is continuous.
    Note that job $j$ (where currently $d_j(t)=\gamma(t)$) uses up the remaining resource.
    When $d_j$ crosses some other dual line $d_{j'}$, $\gamma$ will either continue to coincide with $d_j$ or with $d_{j'}$: This is because $j$ and $j'$ will together still use up the remaining resource.
    An analogous argument applies if multiple dual lines cross at a common point, then $\gamma$ continues as one of them.
    As such $\gamma$ is continuous.
    Trivially, there exists point $t_\gamma\ge 0$ such that $\gamma(t)=0$ for all $t\ge t_{\gamma}$.
    2.
    This follows from the same argument as above:
    As long as no dual line crosses the dual line $d_j$ corresponding to $\gamma$, dual lines above (below) $d_j$ stay above (below) $d_j$, respectively, and the respective job's resource assignment stays constant.
    Therefore, when we have $R_j(t-\delta)\ne R_j(t)$ for some $t\ge 0$ and arbitrarily small $\delta>0$,
    then $j$'s dual line must cross $\gamma$ exactly at $t$, so $d_j(t)=\gamma(t)$.
\end{proof}

We are now ready to prove \Cref{lem:no_overlap, lem:alpha_shapes_composed_of_other_shapes, lem:shape_equalities}.

\begin{proof}[Proof of \Cref{lem:no_overlap}]
    Suppose that $V=\Alpha^{i,l}$ and $W=\Alpha^{i',l'}$ with $V\ne W$ overlap.
    Both contain only points $(t,r,\alpha)$ with $(t,r)\in \Omega^i$ ($\in \Omega^{i'}$), respectively. Therefore we must have $i=i'$ for them to overlap. However, $T^{i,l}$ and $T^{i,l'}$ are disjoint unless $l=l'$. Therefore $V=W$ with a contradiction. For the second statement, if both pieces are $\Rho$-pieces, the argument is analogous.
    Next, note that if $V=\Beta^{i,l}$ and $W=\Gamma^{i,l}$, they cannot overlap by the conditions
    $\alpha\ge \gamma(t)$ and $\alpha<\gamma(t)$, respectively. As such we continue by considering the unions 
    $U^{i,l}=\Beta^{i,l}\cup \Gamma^{i,l}=\set{(t,r,\alpha)|r\in\Omega_{j^{i,l}}^{i}(t^{i,l-1}),\gamma(t^{i,l})\le \alpha< \min(\gamma(t^{i,l-1}),d_{j^{i,l}}(t))}$ instead of the individual $\Beta$- and $\Gamma$-pieces.
    Specifically, we assume $V$ or $W$ to be $U^{i,l}$ instead of $\Beta^{i,l}$ or $\Gamma^{i,l}$.

    Suppose $V=U^{i,l}$ and $W=U^{i',l'}$.
    Then there exists some $(t,r,\alpha)\in U^{i,l}\cap U^{i',l'}$.
    Similarly as above, $(t,r)\in \Omega^{i}_{j}\cap \Omega^{i'}_{j'}$ for some $j,j'\in J$, so $i=i'$ must hold.
    W.l.o.g. let $l>l'$.
    Then $\alpha < \min(\gamma(t^{i,l-1}),d_{{j}^i_{l}}(t))\le\gamma(t^{i,l-1})\le \gamma(t^{i,l'})\le \alpha$,
    where the second inequality follows from $l>l'$ and the second statement of \Cref{obs:completion_properties}.
    This is a contradiction, so $l=l'$ and as such $V=W$.

    Let w.l.o.g. $V=\Rho^{i,l}$ and $W=U^{i',l'}$.
    Assuming they overlap, we have some
    $(t,r,\alpha)\in \Rho^{i,l}\cap U^{i',l'}$.
    For the same reasons as above, $i=i'$ must hold.
    Assume that $l=l'$.
    Then $d_{j^{i,l}}(t) \le \alpha < \min(\gamma(t^{i,l-1}),d_{j^{i,l}}(t))\le d_{j^{i,l}}(t)$ with a contradiction. So we must have $l\ne l'$.
    Let $j=j^{i,l}$ and $j'=j^{i',l'}$.
    Assume $l<l'$.
    Then the two pieces are separated over the priority axis, i.e., $\alpha\ge d_{j}(t)\ge d_{j}(t^{i,l})=\gamma(t^{i,l})\ge \gamma(t^{i,l'})\ge \min(\gamma(t^{i,l'}),d_{j'}(t^{i,l'}))>\alpha$.
    Assume instead $l>l'$.
    Using the second statement of \Cref{obs:completion_properties},
    we must then have $t<t^{i,l'}$, as otherwise
    $\gamma(t^{i,l'})\le \alpha < \min(\gamma(t^{i,l'-1}),d_{j^{i,l'}}(t))\le \min(\gamma(t^{i,l'-1}),d_{j^{i,l'}}(t^{i,l'}))= \min(\gamma(t^{i,l'-1}),\gamma(t^{i,l'}))\le \gamma(t^{i,l'})$.
    Since $t\in T^{i,l}$, we have $t\ge t^{i,l-1}\ge t^{i,l'}>t$ with a contradiction.
\end{proof}

\begin{proof}[Proof of \Cref{lem:alpha_shapes_composed_of_other_shapes}]
    We show two directions, starting with $\Alpha^{\mathrm{all}}\subseteq \Rho^{\mathrm{all}}\cup \Beta^{\mathrm{all}}\cup\Gamma^{\mathrm{all}}$.
    Let $(t,r,\alpha)\in \Alpha^{\mathrm{all}}$, specifically, $(t,r,\alpha)\in \Alpha^{i,l}$ for some $i\in [k], l\in [K^i]$. Denote $j\coloneqq j^{i,l}$.
    If $\alpha\ge d_{j^{i,l}}(t)$, then $(t,r,\alpha)\in \Rho_j\subseteq \Rho^{\mathrm{all}}$.
    If $\gamma(t^{i,l})\le\alpha<d_{j^{i,l}}(t)$, then since $t\in T^{i,l}$,
    $\gamma(t^{i,l})\le\alpha<\min(\gamma(t^{i,l-1}),d_{j^{i,l}}(t))$ and $t\ge 0$.
    Also, $r\in \Omega_{j}^i(t^{i,l-1})$ and as such $(t,r,\alpha)\in \Beta^{i,l}\cup \Gamma^{i,l}\subseteq \Beta^{\mathrm{all}}\cup \Gamma^{\mathrm{all}}$.
    Otherwise, $\alpha<\gamma(t^{i,l})$.
    In such a case there exists
    some $l'>l$ such that $\alpha\in [\gamma(t^{i,l'-1}), \gamma(t^{i,l'}))$.
    This is because of the second property of \Cref{obs:completion_properties}.
    In the time interval $T^{i,l'}$,
    the full resource must still be used as $\gamma$ is still positive.
    Therefore, there exist $\Beta$- and $\Gamma$-pieces $\Beta^{i,l'}$ and $\Gamma^{i,l'}$.
    To show that $(t,r,\alpha)$ lies in one of them, it only remains to show that $\alpha<d_{j^{i,l'}}(t)$.
    This follows from the fact that
    $d_{j^{i,l'}}(t^{i,l'-1})=\gamma(t^{i,l'-1})$ by \Cref{def:omega_areas},
    $t\le t^{i,l'-1}$ (follows from $t\le t^{i,l}$ and $l'>l$) and the fact that $d_{j^{i,l'}}$ is monotonically decreasing.
    This shows $\Alpha^{\mathrm{all}}\subseteq \Rho^{\mathrm{all}}\cup \Beta^{\mathrm{all}}\cup\Gamma^{\mathrm{all}}$.

    Now let $(t,r,\alpha)\in \Rho^{\mathrm{all}}\cup \Beta^{\mathrm{all}}\cup\Gamma^{\mathrm{all}}$.
    First assume that $(t,r,\alpha)\in\Rho^{i,l}$ for some $i\in [k]$, $l\in [K^i]$.
    Then $(t,r)\in \Omega_{j^{i,l}}$ and $t\in T^{i,l}$ as well as $\alpha<\alpha_{j^{i,l}}$, which implies $(t,r,\alpha)\in \Alpha^{\mathrm{all}}$.
    On the other hand, if $(t,r,\alpha)\in\Beta^{i,l}\cup \Gamma^{i,l}$ for some $i\in [k]$, $l\in [K^i]$, then
    $r\in \Omega_{j^{i,l}}(t^{i,l-1})$
    %interfacing with this $r\in \Omega_{j^{i,l}}(t^{i,l-1})$ condition is not great.
    and $\alpha<d_{j^{i,l}}(t)\le \alpha_{j^{i,l}}$.
    This shows that $\Rho^{\mathrm{all}}\cup \Beta^{\mathrm{all}}\cup\Gamma^{\mathrm{all}}\subseteq \Alpha^{\mathrm{all}}$.
\end{proof}

\begin{proof}[Proof of \Cref{lem:shape_equalities}]
    We show the following more detailed statements for all $j\in J$, $i\in [k]$ and $l\in [K^i]$, from which the statements in the lemma follow. Note again that we are using the general $CLP(\bar{v})/DCP(\bar{v})$ definitions (see \Cref{subsec:generalized_clp_dcp}).
    \begin{multicols}{2}
        \begin{enumerate}[nosep, leftmargin=*]
            \item {
                $\abs{\Rho^{i,l}}=q^i \int_{t^{i,l-1}}^{t^{i,l}}t/v_j\dt=\abs{\Beta^{i,l}}+\abs{\Gamma^{i,l}}$.
            }
            \item {
                $\abs{\Rho_j}=\int_{0}^{\infty} R_j(t)\cdot t/v_j\dt$
            }
            \item {
                $\abs{\Alpha_j}=\alpha_j \bar{v}_j$
            }
            \item {
                $\abs{\Gamma^{i,l}}=q^i \int_{\gamma(t^{i,l})}^{\gamma(t^{i,l-1})}\gamma^{-1}(\alpha)\dalpha$
            }
            \item {
                $\abs{\Psi^{\mathrm{all}}}=\Psi$ for all $\Psi\in\set{\Rho,\Alpha,\Beta,\Gamma}$
            }
        \end{enumerate}
    \end{multicols}

    These statements are proven as follows.
    
    1.
    By definition, $\Rho^{i,l}$ is a shape that is $q^i$ deep into the resource axis,
    and its front face is a (possibly degenerate) trapezoid.
    At time $t\in T^{i,l}$, it has a height of
    $\max(0,\alpha_j-(\alpha_j-t/v_j))=t/v_j$, where $j\coloneqq j^{i,l}$.
    As such $\abs{\Rho^{i,l}}=q^i \int_{t^{i,l-1}}^{t^{i,l}} t/v_j\dt$.

    Next, consider the union of the $\Beta$- and $\Gamma$-pieces
    \[
    \Beta^{i,l}\cup \Gamma^{i,l}=\set{(t,r,\alpha)|r\in\Omega_{j^{i,l}}^{i}(t^{i,l-1}),\gamma(t^{i,l})\le \alpha< \min(\gamma(t^{i,l-1}),d_{j^{i,l}}(t)),t\ge 0}.
    \]
    By \Cref{lem:no_overlap}, the two pieces do not overlap, as such the union has a volume of
    $\abs{\Beta^{i,l}}+\abs{\Gamma^{i,l}}$.
    Again, the shape is $q^i$ deep into the resource axis.
    Its front face is a trapezoid with a height of
    $\max(0,\min(\gamma(t^{i,l-1}),d_{j^{i,l}}(t))-\gamma(t^{i,l}))$
    at time $t$.
    We can then calculate
    \begin{align*}
        \abs{\Beta^{i,l}}+\abs{\Gamma^{i,l}}
        &=q^i\int_{0}^{\infty} \max(0,\min(\gamma(t^{i,l-1}),d_j(t))-\gamma(t^{i,l}))\dt\\
        &=q^i\left(\int_{0}^{t^{i,l-1}} \gamma(t^{i,l-1})-\gamma(t^{i,l})\dt+\int_{t^{i,l-1}}^{t^{i,l}} d_j(t)-\gamma(t^{i,l})\dt\right)\\
        &=q^i\left(\int_{0}^{t^{i,l-1}} d_j(t^{i,l-1})-d_j(t^{i,l})\dt+\int_{t^{i,l-1}}^{t^{i,l}} d_j(t)-d_j(t^{i,l})\dt\right)\\
        &=q^i\left(\int_{0}^{t^{i,l-1}} \alpha_j-\frac{t^{i,l-1}}{v_j}-\alpha_j+\frac{t^{i,l}}{v_j}\dt+\int_{t^{i,l-1}}^{t^{i,l}} \alpha_j-\frac{t}{v_j}-\alpha_j+\frac{t^{i,l}}{v_j}\dt\right)\\
        &=\frac{q^i}{v_j}\left(\int_{0}^{t^{i,l-1}} -t^{i,l-1}+t^{i,l}\dt+\int_{t^{i,l-1}}^{t^{i,l}} t^{i,l}\dt-\int_{t^{i,l-1}}^{t^{i,l}} t\dt\right)\\
        &=\frac{q^i}{v_j}\left((t^{i,l}+t^{i,l-1})(t^{i,l}-t^{i,l-1})-\frac12\left({t^{i,l}}^2-{t^{i,l-1}}^2\right)\right)
        =q^i \int_{t^{i,l-1}}^{t^{i,l}} \frac{t}{v_j}\dt.
    \end{align*}

    2. We show $\abs{\Rho_j}=\int_{0}^{\infty} R_j(t)\cdot t/v_j\dt$:
    \begin{align*}
        \abs{\Rho_j}
        &=\abs{\bigcup_{i\in [k],l\in [K^i],j^{i,l}=j} \Rho^{i,l}}
        =\sum_{i\in [k],l\in [K^i],j^{i,l}=j} \abs{\Rho^{i,l}}
        =\sum_{i\in [k],l\in [K^i],j^{i,l}=j} \int_{t^{i,l-1}}^{t^{i,l}}q^i \frac{t}{v_j}\dt\\
        &=\sum_{i\in [k]} \int_{0}^{\infty}q^i\cdot \mathds{1}_{(t,.)\in\Omega^i_j} \frac{t}{v_j}\dt
        =\int_{0}^{\infty}R_j(t) \frac{t}{v_j}\dt
    \end{align*}

    3. The volume of any $\Alpha$-piece $\Alpha^{i,l}$ is by definition
    $\abs{\Alpha^{i,l}}=q^i \alpha_{j^{i,l}} \abs{T^{i,l}}$. Then
    \begin{align*}
        \abs{\Alpha_j}
        &=\abs{\bigcup_{i\in [k],l\in [K^i],j^{i,l}=j} \Alpha^{i,l}}
        =\sum_{i\in [k],l\in [K^i],j^{i,l}=j} \abs{\Alpha^{i,l}}\\
        &=\sum_{i\in [k],l\in [K^i],j^{i,l}=j} q^i \alpha_j \abs{T^{i,l}}
        =\alpha_j\sum_{i\in [k]} \abs{\Omega^i_j}
        =\alpha_j \bar{v}_j
    \end{align*}

    4.
    The inequalities describing $\Gamma^{i,l}$ are
    $\gamma(t^{i,l})\le \alpha<\min(\gamma(t^{i,l-1},d_{j^{i,l}}(t))),\alpha<\gamma(t)$.
    They can be simplified to
    $\gamma(t^{i,l})\le \alpha<\min(\gamma(t^{i,l-1}), \gamma(t))$ as
    the job $j^{i,l}$ scheduled must always satisfy $d_{j^{i,l}}(t)\ge \gamma(t)$ to be scheduled.
    Therefore, $\Gamma^{i,l}$ contains all points $(t,r,\alpha)$ with
    $r\in \Omega^{i}_{j^{i,l}}$, $\alpha$ lies between
    $\gamma_{t^{i,l}}$ and $\gamma(t^{i,l-1})$, and $(t,\alpha)$ must be a point to the left of the fracture line $\gamma$ (since $\alpha\le \gamma(t)$).
    This can be described by the integral over
    Since $\gamma$ is strictly monotonically decreasing in the interval $\intco{t^{i,l-1},t^{i,l}}$ (see \Cref{obs:completion_properties}), we can express $\abs{\Gamma^{i}_l}$ by an integral over its inverse as
    \begin{align*}
        \abs{\Gamma^i_l}=q^i \int_{\gamma(t^{i,l})}^{\gamma(t^{i,l-1})} \gamma^{-1}(\alpha)\dalpha
    \end{align*}
    (By \Cref{obs:completion_properties}, statement 2, we know that $\gamma$ is invertible over the interval $\intco{0,t_\gamma}$ for some $t_\gamma\ge 0$ and as such also over $\intco{t^{i,l-1},t^{i,l}}$.

    5.
    For $\Psi\in\set{\Rho,\Alpha}$, we can calculate
    \begin{align*}
        \abs{\Psi^{\mathrm{all}}}=\abs{\bigcup_{i\in [k],l\in [K^i]}\Psi^{i,l}}=\abs{\bigcup_{j\in J}\bigcup_{i\in [k],l\in [K^i],j=j^{i,l}}\Psi^{i,l}}=\abs{\bigcup_{j\in J}\Psi_j}=\sum_{j\in J}\abs{\Psi_j}=\sum_{j\in J}\dots=\Psi,
    \end{align*}
    where for \glqq$\dots$\grqq, we can insert the results from statements 2 and 3.

    Consider $\Psi=\Beta$.
    We have to show $\abs{\Beta_j}=r_j \int_{0}^\infty \beta_j(t)\dt$, since this will allow us to derive the statement:
    \begin{align*}
        \abs{\Beta^{\mathrm{all}}}
        &=\abs{\bigcup_{i\in [k],l\in [K^i]}\Beta^{i,l}}
        =\abs{\bigcup_{j\in J}\bigcup_{i\in [k],l\in [K^i],j=j^{i,l}}\Beta^{i,l}}\\
        &=\abs{\bigcup_{j\in J} \Beta_j}
        =\sum_{j\in J} \abs{\Beta_j}
        =\sum_{j\in J}\int_{0}^{\infty}\beta_j(t)\dt
        =\Beta
    \end{align*}

    So consider any $\Beta_j$.
    It consists of pieces $\Beta^{i,l}$ with $j^{i,l}=j$.
    $\Beta^{i,l}$ is $q^i$ deep into the resource axis.
    Further the points $(t,r,\alpha)\in\Beta^{i,l}$ satisfy
    $\max(\gamma(t^{i,l}),\gamma(t))\le \alpha<\min(\gamma(t^{i,l-1}),d_{j}(t))$.
    Using that $\gamma(t)\le \alpha<\gamma(t^{i,l-1})$ (and the fact that $\gamma$ is monotonically decreasing (see \Cref{obs:completion_properties}), we derive $t^{i,l-1}\le t$
    From $\gamma(t^{i,l})\le \alpha<d_{j}(t)$ follows that $t<t^{i,l}$, since
    by definition of a geometrical representation, we have that $\gamma(t^{i,l})=d_{j}(t^{i,l})$.
    As such, $t\in T^{i,l}$.

    The height of the piece at time $t\in T^{i,l}$ is then $d_{j}(t)-\gamma(t)=\beta_j(t)$. (The last equality follows from the $R$-slackness condition.)
    As such the volume of the piece is
    $\abs{\Beta^{i,l}}=q^i \int_{t^{i,l-1}}^{t^{i,l}} \beta_j(t)\dt$.
    Because of the $\beta$-slackness condition, $\beta_j(t)>0$ iff $j$ is assigned $r_j$ resource at time $t$. As such, when we sum
    over all strip widths $q^i$ where $j$ is scheduled, we will obtain exactly $r_j$.
    As such we can finish our calculation:
    \begin{align*}
        \abs{\Beta_j}
        &=\abs{\bigcup_{i\in [k],l\in [K^i],j=j^{i,l}}\Beta^{i,l}}
        =\sum_{i\in [k],l\in [K^i],j=j^{i,l}} \abs{\Beta^{i,l}}
        =\sum_{i\in [k]}\sum_{l\in [K^i],j=j^{i,l}} \int_{t^{i,l-1}}^{t^{i,l}} q^i \beta_j(t)\dt\\
        &=\sum_{i\in [k]} \int_{0}^{\infty} q^i\cdot \mathds{1}_{(t,.)\in \Omega^i_j} \beta_j(t)\dt
        =\int_{0}^{\infty} r_j \beta_j(t)\dt
    \end{align*}

    Lastly, for $\Psi=\Gamma$, we calculate
    \begin{equation*}
        \hspace{0.5em}
        \abs{\Gamma^{\mathrm{all}}}
        =\abs{\bigcup_{i\in [k],l\in [K^i]}\Gamma^{i,l}}
        =\sum_{i\in [k],l\in [K^i]}\abs{\Gamma^{i,l}}
        =\sum_{i\in [k]}q^i\sum_{l\in [K^i]} \int_{\gamma(t^{i,l})}^{\gamma(t^{i,l-1})} \gamma^{-1}(\alpha)\dalpha
    \end{equation*}
    \begin{equation*}
        \hspace{-0.5em}
        =\sum_{i\in [k]}q^i\int_{\gamma(t_{\gamma})}^{\gamma(0)} \gamma^{-1}(\alpha)\dalpha
        =\sum_{i\in [k]}q^i\int_{0}^{t_{\gamma}} \gamma(t)\dt
        =\int_{0}^{\infty} \gamma(t)\dt
        =\Gamma\qedhere
    \end{equation*}
\end{proof}

\section{A $(3/2 + \varepsilon)$-Approximation for Total Completion Time Minimization}
\label{sec:threehalfspluseps}

\subsection{Algorithm Description}
\label{subsec:algorithm_desc_three_halfs_plus_eps}

Our $(3/2+\varepsilon)$-approximation follows the same principle as described in \Cref{subsec:mainpart_algorithm_description}, but uses algorithm \textsc{LSApprox} instead of \textsc{LS}.
Algorithm \textsc{Greedy} stays the same.

We give a broad description of \textsc{LSApprox}.
It first computes an approximate solution to the $CLP$ (see \Cref{def:CLP}).
For that, the job set is first subdivided (see \Cref{def:job_subdivision} for the formal definition).
Then a small portion of the total resource (say $\mu=\kappa\cdot \varepsilon\in \intoo{0, 1}$, $1/\mu\in\N$ for some constant factor $\kappa>0$) is reserved for jobs with a small resource requirement and/or processing time.
For the remaining jobs, the dual $LP$ (see \Cref{par:time_discretized_lp_and_its_dual}) is solved optimally with a polynomial number of slots.
The $\alpha$-values of this solution are then used to construct a line schedule
$(R,\alpha,\cdot,\cdot,\bar{v})$ for some $\bar{v}$.
(It is clear that this can be done in polynomial time, as it effectively boils down to calculating the intersections of the dual lines.)
Because generally $\bar{v}\ne v$ will hold, the schedule is scaled horizontally afterwards, i.e., a new schedule $\tilde{R}(t)= R(t/s)$ is created, where $s=\max\set{v_j/\bar{v}_j|j\in J}$.
Lastly, $\tilde{R}$ is squashed to only use $1-\mu$ overall resource (scale vertically by $1-\mu$ and horizontally by $1/(1-\mu)$), and packed together with a schedule for the small jobs that used $\mu$ resource.
We denote the schedule produced by \textsc{LSApprox} by $R^{FA}$.

In principle, we could just squash this $LP$ solution for these remaining jobs to only use
the remaining $(1-\mu)$ resource and would obtain an FPTAS for the fractional problem.
The issue with this approach is, however, that the $LP$/dual $LP$ solutions do not necessarily correspond to line schedules, and we therefore could not use our knowledge about these.
This is why we reuse the $\alpha$-values from the dual LP
to produce a line schedule.
We then show that the scheduled volumes $\bar{v}$ are close to the actual volumes $v$ (in dependence on the granularity of the slots). By also choosing $\mu$ to be very small, we can guarantee that the overall error is small enough.

Analogous to \Cref{prop:fractional_bound}, we will prove the following \lcnamecref{lem:approximation}.

\begin{proposition}
    \label{lem:approximation}
    For any $\varepsilon>0$,
    \textsc{LSApprox} produces in polynomial time in $n,1/\varepsilon$ a schedule $R^{FA}$ with $\gamma(R^{FA})\le 2+\varepsilon$.
\end{proposition}

Using this \lcnamecref{lem:approximation}, we can show our main \lcnamecref{thm:three_halfs_plus_eps}.

\thmthreehalfspluseps*

\begin{proof}
    The proof follows the same idea as the proof of \Cref{lem:approx}.
    We run \textsc{Greedy} and \textsc{LSApprox} and choose the schedule with the smaller total completion time.
    Both algorithms run in polynomial time in
    $n, 1/\varepsilon$
    (\textsc{LSApprox} by \Cref{lem:approximation}).
    Plugging \Cref{lem:approximation} into \Cref{lem:approx}, we obtain an approximation ratio of $(3+\varepsilon)/2\le 3/2+\varepsilon$.
    \iffalse
    \begin{equation*}
        \hspace{-6.5em}
        \min(C(R^{FA}),C(R^G))
        \le \min((2+\varepsilon) C^{F*},C^A+C^L)
        \le \min((2+\varepsilon) (C^*- C^L/2),C^*+C^L)
    \end{equation*}
    \begin{equation*}
        \hspace{9.8em}
        =\frac{3+\varepsilon}{2}\cdot C^*-\frac{\varepsilon}{4}\cdot C^L+\min\left(\frac{1+\varepsilon}{2}\cdot C^*-\left(1+\frac{\varepsilon}{4}\right)\cdot C^L, -\left(\frac{1+\varepsilon}{2}\cdot C^*-\left(1+\frac{\varepsilon}{4}\right)\cdot C^L\right)\right)
    \end{equation*}
    \begin{equation*}
        \hspace{-9.6em}
        \le \frac{3+\varepsilon}{2}\cdot C^*-\frac{\varepsilon}{4}\cdot C^L
        \le \left(\frac32+\varepsilon\right) C^*\qedhere
    \end{equation*}
    \fi
\end{proof}

The remainder of this section is dedicated to the proof of \Cref{lem:approximation}.
The proof is based on the following \Cref{lem:error_subdivision_and_squash, lem:slotting_instead_of_CLP_error, lem:error_continuous_instead_of_slotted}.
These propositions use a subdivision $J$ into three job sets.
These are $J^{\mathrm{l}}$ (\emph{light} jobs: jobs with small resource requirement), $J^{\mathrm{sh}}$ (\emph{short-heavy} jobs: jobs with large resource requirement but small processing time) and $J^{\mathrm{lh}}$ (\emph{long-heavy} jobs: both large resource requirement and processing time).
For details, see the below \Cref{def:job_subdivision}.

\Cref{lem:error_subdivision_and_squash} essentially allows us to ignore the light and short-heavy jobs for the remainder of analysis.
\Cref{lem:slotting_instead_of_CLP_error} shows that, 
with a fine enough slotting, optimal $LP$ solution (and therefore also the optimal $DP$ solution) closely approximates the optimal $CLP$ solution.
Lastly, \Cref{lem:error_continuous_instead_of_slotted} deals with the step of \textsc{LSApprox} where we translate the optimal $\alpha$-values into a line schedule.
It guarantees that the volume scheduled by the line schedule is close to the actual job volumes and that the cost of the produced line schedule is close to the optimal dual $LP$ cost.

\begin{proposition}
    \label{lem:error_subdivision_and_squash}
    Let $R^{\mathrm{lh}}$ be a feasible schedule for $J^{\mathrm{lh}}$ and
    $R^*$ an optimal schedule for $J$.
    Then we can compute in polynomial time a feasible schedule $R^{FA}$ for $J$ such that $C(R^{FA})\le (1+2\mu)\cdot (C(R^{\mathrm{lh}}) + \sum_{j\in J^{\mathrm{l}}}C_j(R^*))$ if $\mu>0$ is sufficiently small.
\end{proposition}

%Christoph: in case you find another c: c=A-5
%Christoph: in case you find another A: A=6

\begin{proposition}
    \label{lem:slotting_instead_of_CLP_error}
    Let $R^S$ be a schedule induced from an optimal (slotted) $LP$ solution for $J^{\mathrm{lh}}$ with a time horizon $T=n\cdot p_{\mathrm{max}}$ and with a slot width of $\delta=T\cdot (\mu/n)^6$,
    and $R^{\mathrm{lh}*}$ be an optimal $CLP$ solution for $J^{\mathrm{lh}}$.
    Then $C^F(R^S)\le (1+2\mu)\cdot C^F(R^{\mathrm{lh}*})$ if $\mu>0$ is sufficiently small.
\end{proposition}

\begin{proposition}
    \label{lem:error_continuous_instead_of_slotted}
    Let $\alpha=(\alpha_j)_{j\in J^{\mathrm{lh}}}$ be the the $\alpha$-values of an optimal dual $LP$ solution
    for job set $J^{\mathrm{lh}}$.
    Let further $(R^{\mathrm{lh}},\alpha,\cdot,\cdot,\bar{v})$ be the line schedule for $\alpha$ and $R^S$ be defined as in \Cref{lem:slotting_instead_of_CLP_error}.
    Then the following statements hold:
    \begin{itemize}
        \item {
            For each $j\in J^{\mathrm{lh}}$, $\bar{v}_j\ge (1-\mu^{4}/n)v_j$.
        }
        \item {
            For small enough $\mu$, $C^F(R^{\mathrm{lh}})\leq (1+5\mu)C^F(R^S)$.
        }
    \end{itemize}
\end{proposition}

With these statements, we are able to prove \Cref{lem:approximation}.

\begin{proof}[Proof of \Cref{lem:approximation}]
    Recall that $R^{\mathrm{lh}}$ comes from a line schedule $(R^{\mathrm{lh}},\cdot,\cdot,\cdot,\bar{v})$.
    By slightly modifying the calculation of \Cref{prop:fractional_bound}, and using a simplification of the first statement of \Cref{lem:error_continuous_instead_of_slotted} (omit a $1/n$ factor) we get
    \setcounter{footnote}{11}
    \footnotetext{\label{ftn:smallereqstarop}We use $\le^*$ to indicate that the inequality only holds for small enough $\mu>0$. The term on the right was obtained by taking the Taylor series and bounding all terms of degree at least $2$ by $\mu$.}
    \begin{align*}
        C(R^{\mathrm{lh}})
        &=\sum_{j\in J} C_j(R^{\mathrm{lh}})
        \le \sum_{j\in J}{\alpha_j v_j}
        \le \sum_{j\in J}{\alpha_j \frac{\bar{v}_j}{1-\mu^4}}\
        =\frac{\Alpha}{1-\mu^4}
        =\frac{\Alpha-\Beta-\Gamma+\Rho}{1-\mu^4}\\
        &=\frac{2\Rho}{1-\mu^4}
        =\frac{2}{1-\mu^4}\cdot C^F(R^{\mathrm{lh}})
        \le^*\text{\footnoteref{ftn:smallereqstarop}} (2+\mu)\cdot C^F(R^{\mathrm{lh}})
    \end{align*}
    
    We can bound the cost of $R^{FA}$ by using the above and \Cref{lem:error_subdivision_and_squash, lem:slotting_instead_of_CLP_error, lem:error_continuous_instead_of_slotted}:
    \begin{align*}
        C(R^{FA})
        &\le (1+2\mu) \left(C(R^{\mathrm{lh}})+\sum_{j\in J^l} C_j(R^*)\right)\\
        &\le (1+2\mu) \left((2+\mu)\cdot C^F(R^{\mathrm{lh}})+\sum_{j\in J^l} C_j(R^*)\right)
        \\
        &\le (1+2\mu) \left((2+\mu) (1+5\mu)C^F(R^S)+\sum_{j\in J^l} C_j(R^*)\right)\\
        &\le (1+2\mu) \left((2+\mu) (1+5\mu)(1+2\mu)C^F(R^{\mathrm{lh}*})+\sum_{j\in J^l} C_j(R^*)\right)
        \\
        &\le (1+2\mu) (2+\mu) (1+5\mu)(1+2\mu) C^{F*}
        \le^* (2+20\mu) C^{F*}
    \end{align*}
    By defining $\mu=\varepsilon/20$, we obtain the desired bound.
    For the running time, note first that
    \textsc{LSApprox} runs in polynomial time in $1/\mu$ and $n$.
    Since $\mu$ and $\varepsilon$ differ only by a multiplicative factor, \textsc{LSApprox} will also run in polynomial time in $1/\varepsilon$.
\end{proof}

\subsection{Dividing the Set of Jobs}

Our goal for this sub-section is to show \Cref{lem:error_subdivision_and_squash} so that we can  essentially ignore all jobs except $J^{\mathrm{lh}}$.
Before proving \Cref{lem:error_subdivision_and_squash}, we give the formal definition of the job subsets.

\begin{definition}
    \label{def:job_subdivision}
    We define the sets of
    \begin{enumerate}
        \item {
            \emph{light} jobs $J^{\mathrm{l}}=\set{j\in J|r_j\le \mu/n}$,
        }
        \item {
            \emph{short-heavy} jobs $J^{\mathrm{sh}}=\set{j\in J|p_j\le (\mu/n)^2\cdot p_{\mathrm{max}}\text{ and } r_j>\mu/n}$ and
        }
        \item {
            \emph{long-heavy} jobs $J^{\mathrm{lh}}=J\setminus (J^{\mathrm{l}}\cup J^{\mathrm{sh}})$
        }
    \end{enumerate}
\end{definition}

As described in \Cref{subsec:algorithm_desc_three_halfs_plus_eps}, we squash $R^{\mathrm{lh}}$ to only use $1-\mu$ resource, and then pack the remaining jobs such that they only use $\mu$ resource in total at each time point.
To explain in more detail,
each job $j\in J_{\mathrm{l}}\cup J_{\mathrm{sh}}$ is assigned a resource of $r^{(j)}\coloneqq \min(\mu/n,r_j)$ until it finishes, while $R^{\mathrm{lh}}$ is squashed such that it uses at most a resource of $1-\mu$.
Formally, define
\begin{align*}
    R^{FA}_j(t)=\begin{cases}
        r^{(j)}\cdot \mathds{1}_{t<v_j/r^{(j)}}&\text{if $j\in J_{\mathrm{l}}\cup J_{\mathrm{sh}}$}\\
        (1-\mu)\cdot R_j^{\mathrm{lh}}(t/(1-\mu))&\text{if $j\in J^{\mathrm{lh}}$.}
    \end{cases}
\end{align*}
The schedule is feasible as all jobs have been scheduled,
and the light and short-heavy jobs use a resource of at most
$n\cdot \min(\mu/n,r_j)\le \mu<1$, so the resource is not overused. Furthermore, $r^{(j)}\le r_j$, so each light and short-heavy job is assigned at most its resource requirement.

\begin{proof}[Proof of \Cref{lem:error_subdivision_and_squash}]
    Obviously, $R^{FA}$ (as described above) can be calculated in polynomial time.
    
    We calculate the cost of the resulting schedule.
    It is easy to calculate that $C_j(R^{FA})=1/(1-\mu)\cdot C_j(R^{\mathrm{lh}})$ for $j\in J^{\mathrm{lh}}$.
    Further, jobs $j\in J^{\mathrm{l}}$ are scheduled from time $0$ with full resource, so they cannot be scheduled better in $R^*$, giving $C_j(R^{FA})\le C^F_j(R^*)$.

    Notice that $C_{j_{\mathrm{max}}}\ge p_{\mathrm{max}}$ for any job $j_{\mathrm{max}}\notin J^{\mathrm{sh}}$ with $p_{j_{\mathrm{max}}}=p_{\mathrm{max}}$ (which must exist since $\mu<1\le n$.
    The total completion time of short-heavy jobs is then
    \begin{align*}
        \sum_{j\in J^{\mathrm{sh}}} C_j(R^{FA})
        &=\sum_{j\in J^{\mathrm{sh}}} \frac{v_j}{\min(\mu/n,r_j)}
        \le \sum_{j\in J^{\mathrm{sh}}} \frac{n p_j r_j}{\mu}\\
        &\le \sum_{j\in J^{\mathrm{sh}}} \frac{r_j \mu p_{\mathrm{max}}}{n}
        \le \mu p_{\mathrm{max}}\le \mu \cdot \max (C(R^{\mathrm{lh}}), C(R^*)),
    \end{align*}
    where the last inequality is due to $j_{\mathrm{max}}\in J^{\mathrm{lh}}\cup J^{\mathrm{l}}$. Therefore,
    \begin{equation*}
        \hspace{-10.2em}
        C(R^{FA})
        =\sum_{j\in J^{\mathrm{l}}} C_j(R^{FA})+\sum_{j\in J^{\mathrm{sh}}} C_j(R^{FA})+\sum_{j\in J^{\mathrm{lh}}} C_j(R^{FA})
    \end{equation*}
    \begin{equation*}
        \hspace{-2.5em}
        \le \sum_{j\in J^{\mathrm{l}}} C_j(R^*)+\mu \cdot \max (C(R^{\mathrm{lh}}), C(R^*))+\frac{1}{1-\mu}\sum_{j\in J^{\mathrm{lh}}} C_j(R^{\mathrm{lh}})
    \end{equation*}
    \begin{equation*}
        \hspace{-9.5em}
        \le \left(\frac{1}{1-\mu}+\mu\right) C(R^{\mathrm{lh}})+(1+\mu)\sum_{j\in J^{\mathrm{l}}} C_j(R^*)
    \end{equation*}
    \begin{equation*}
        \hspace{-10.5em}
        \le \left(\frac{1}{1-\mu}+\mu\right)\left(C(R^{\mathrm{lh}})+\sum_{j\in J^{\mathrm{l}}} C_j(R^*)\right)
    \end{equation*}
    \begin{equation*}
        \hspace{-12.0em}
        \le^* (1+2\mu)\left(C(R^{\mathrm{lh}})+\sum_{j\in J^{\mathrm{l}}} C_j(R^*)\right)\qedhere
    \end{equation*}
\end{proof}

\subsection{Approximating the $CLP$ with an $LP$ Solution}

This sub-section will be concerned with the proof of \Cref{lem:slotting_instead_of_CLP_error}.
For this we consider optimal $LP$ solutions using the parameters described in the \lcnamecref{lem:slotting_instead_of_CLP_error}, specifically
a time horizon $T=n\cdot p_{\mathrm{max}}$ and a slot width of $\delta=T\cdot (\mu/n)^6$.
Using these parameters, we can first make the following observation for the schedule $R^S$ for the long-heavy jobs $J^{\mathrm{lh}}$.

\begin{observation}
    \label{obs:job_scheduled_in_many_slots}
    Consider a schedule $R^S$ corresponding to an optimal $LP$ solution for $J^{\mathrm{lh}}$, and let $\alpha$ be the corresponding $\alpha$-values for the optimal dual $LP$ solution.
    Then the following statements hold.
    \begin{enumerate}
        \item {
            Each $j\in J^{\mathrm{lh}}$ can at most process $\mu^{4}/n^{3}\cdot v_j$ volume in each slot.
        }
        \item {
            The resource assignment of $R^S$ and the schedule
            $R^{\mathrm{lh}}$ from the line schedule $(R^{\mathrm{lh}}, \alpha, \cdot, \cdot, \cdot)$ of $\alpha$ differ in at most $n^2$ slots.
        }
    \end{enumerate}
\end{observation}

\begin{proof}
    1. In each slot, a job $j\in J^{\mathrm{lh}}$ can at most process a volume of
    $r_j \delta=r_j T\cdot (\mu/n)^6=r_j \cdot \mu^6/n^{5}\cdot p_{\mathrm{max}}<r_j \cdot \mu^6/n^{5}\cdot p_j n^2/\mu^2=\mu^{4}/n^{3}\cdot v_j$, where the inequality comes from the definition of long-heavy jobs.

    2.
    The two schedules may only differ in slots where two dual lines (using vector $\alpha$) intersect with each other or with the time axis.
    There are $\abs{J^{\mathrm{lh}}}\le n$ dual lines, so the number of intersections between two dual lines
    is $n(n-1)/2$.
    Additionally, there are $n$ intersections of a dual line with the time axis.
    In the worst-case, each intersection lies in a different slot and the intersections happen above the time axis, so there are at most $n(n-1)/2+n\le n^2$
    slots that contain dual line intersections.
\end{proof}

Our general proof idea for 
\Cref{lem:slotting_instead_of_CLP_error} is to convert
an optimal fractional solution $R^{\mathrm{lh}*}$ for $J^{\mathrm{lh}}$ into a slotted solution $R^S$ and to observe how the costs of the two schedules relate.
The conversion happens by averaging the resource assignment
of $R^{\mathrm{lh}*}$ in each slot to obtain the slotted solution.
This averaging mostly does not change the cost too much because the starting and endpoints of a slot are often close enough together that rescheduling volume inside of it does not make a big difference to the overall objective.
However, in the first few slots, this argument breaks down.
This is why we first show that the cost of the first few slots is negligible (see \Cref{lem:slotting_instead_of_CLP_error_helper}).
Afterwards, we can show \Cref{lem:slotting_instead_of_CLP_error}.
We denote by $f\coloneqq 3n^{3}/\mu^{1}$ the number of slots at the beginning of the schedule we neglect.

\begin{lemma}
    \label{lem:slotting_instead_of_CLP_error_helper}
    Let $R^S$ be a schedule induced from an optimal (slotted) $LP$ solution for $J^{\mathrm{lh}}$ with a time horizon $T=n\cdot p_{\mathrm{max}}$ and with a slot width of $\delta=T\cdot (\mu/n)^6$.
    Let $C^F_{\mathrm{first}}$ be the contribution of the first $f$ slots to the fractional total completion time of $R^S$, and $C^F_{\mathrm{last}}$ be the contribution of the other slots,
    Then $C^F_{\mathrm{first}}\le \mu^{4}/n^{3}\cdot f\cdot \frac{1}{1-\mu}\cdot C^F_{\mathrm{last}}$.
\end{lemma}

\begin{proof}
    By \Cref{obs:job_scheduled_in_many_slots}, we know that
    each job can at most schedule $\mu^{4}/n^{3} \cdot v_j$ volume in each slot.
    This means that the first $f$ slots will process at most $f\cdot \mu^{4}/n^{3} \cdot v_j$ processing volume, and consequently, the other slots will process at least $(1-f\cdot \mu^{4}/n^{3}) v_j\le (1-\mu) v_j$ processing volume.
    Using this, we can show that the cost of the first $f$ slots is negligible:
    \begin{align*}
        C^F_{\mathrm{first}}&=\frac{1}{v_j} \sum_{i\le f} V_{j,i}\left(i\delta-\frac{\delta}{2}\right)
        \le \frac{1}{v_j} \sum_{i\le f} V_{j,i}\left(f\cdot \delta-\frac{\delta}{2}\right)
        =\frac{1}{v_j} \left(f\cdot \delta-\frac{\delta}{2}\right) \sum_{i\le f} V_{j,i}
        \\
        &\le \frac{1}{v_j} \left(f\cdot \delta-\frac{\delta}{2}\right) f\cdot \frac{\mu^{4}}{n^{3}} v_j
        = \frac{\mu^{4}}{n^{3}} f\frac{1}{1-\mu}\frac{1}{v_j} \left(f\cdot \delta-\frac{\delta}{2}\right) \left(1-\mu\right) v_j
        \\
        &\le \frac{\mu^{4}}{n^{3}} f\frac{1}{1-\mu}\frac{1}{v_j} \left(f\cdot \delta-\frac{\delta}{2}\right) \sum_{i\ge f} V_{j,i}
        = \frac{\mu^{4}}{n^{3}} f\cdot \frac{1}{1-\mu}\frac{1}{v_j}\sum_{i> f} V_{j,i}\left(f\cdot \delta-\frac{\delta}{2}\right)
    \end{align*}
    \begin{equation*}
        \hspace{-2.9em}
        \le \frac{\mu^{6}}{n^{3}} f\frac{1}{1-\mu}\left(\frac{1}{v_j}\sum_{i> f} V_{j,i}\left(i\delta-\frac{\delta}{2}\right)\right)
        = \frac{\mu^{4}}{n^{3}} f\cdot \frac{1}{1-\mu}\cdot C^F_{\mathrm{last}}.\qedhere
    \end{equation*}
\end{proof}

\begin{proof}[Proof of \Cref{lem:slotting_instead_of_CLP_error}]
    First note that $R^{\mathrm{lh}*}$ will not schedule past the time horizon $T$ by definition of $T$.
    We show that $C^F(R^S)\le (f+1)/f\cdot C^F(R^{\mathrm{lh}*})$ by providing a slotted solution
    based on $R^{\mathrm{lh}*}$.
    For that, we set the $LP$ variables $V_{j,i}=\int_{i\delta}^{(i+1)\delta} R^{\mathrm{lh}*}_j(t)\dt$.
    Call the $LP$ solution $V$.

    Denote by $C^F_j(V)$ the fractional contribution of job $j$ in the LP in solution $V=(V_{j,i})_{j\in J,i\in I}$, i.e., $C^F_j(V)=1/v_j\cdot\sum_{i\in I} V_{j,i}\left(i\delta-\delta/2\right)$.
    We show the cost of $V$ using \Cref{lem:slotting_instead_of_CLP_error_helper}.
    The cost of $V$ for a job $j$ is then as follows.
    \begin{align*}
        C^F_j(V)
        &=\frac{1}{v_j}\sum_{i\in I} V_{j,i}\left(i\delta-\frac{\delta}{2}\right)
        =\frac{1}{v_j} \sum_{i\le f} V_{j,i}\left(i\delta-\frac{\delta}{2}\right)+\frac{1}{v_j}\sum_{i> f} V_{j,i}\left(i\delta-\frac{\delta}{2}\right)
        \\
        &\le \left(1+\frac{\mu^{4}}{n^{3}} f\cdot \frac{1}{1-\mu}\right)\frac{1}{v_j}\sum_{i> f} V_{j,i}\left(i\delta-\frac{\delta}{2}\right)
    \end{align*}
    We can then bound the right hand side $RH\coloneqq \frac{1}{v_j}\sum_{i> f} V_{j,i}\left(i\delta-\frac{\delta}{2}\right)$ as follows.
    We temporarily omit the prefactor and continue the calculation:
    \begin{align*}
        RH
        &\le \frac{1}{v_j}\sum_{i> f} \frac{i}{i-1}\cdot (i-1)\delta V_{j,i}
        \le \frac{f+1}{f}\frac{1}{v_j}\sum_{i> f} (i-1)\delta V_{j,i}\\
        &\le \frac{f+1}{f} \frac{1}{v_j}\sum_{i> f} \left((i-1)\delta
        V_{j,i}+\frac{{V_{j,i}}^2}{2 r_j}\right)
        \\
        &=\frac{f+1}{f}\frac{1}{v_j}\sum_{i> f} \frac{r_j}2\left(\left((i-1) \delta+\frac{V_{j,i}}{r_j}\right)^2-((i-1)\delta)^2\right)
        \\
        &=\frac{f+1}{f}\frac{1}{v_j}\sum_{i> f} \int_{(i-1)\delta}^{(i-1)\delta+V_{j,i}/r_j}{r_j\cdot t\dt}
        \le \frac{f+1}{f}\frac{1}{v_j}\sum_{i> f} \int_{(i-1)\delta}^{i\delta}{R^{\mathrm{lh}*}_j(t)\cdot t\dt}
        \\
        &\le \frac{f+1}{f}\frac{1}{v_j}\sum_{i\in I} \int_{(i-1)\delta}^{i\delta}{R^{\mathrm{lh}*}_j(t)\cdot t\dt}
        = \frac{f+1}{f}\int_{0}^{\infty}{\frac{R^*_j(t)\cdot t}{v_j}\dt}
        = \frac{f+1}{f} C_j^F(R^{\mathrm{lh}*})
    \end{align*}
    In total we get
    \begin{align*}
        C_j^F(V)&\le \frac{f+1}{f} \left(1+\frac{\mu^{4}}{n^{3}} f\cdot \frac{1}{1-\mu}\right)\cdot C_j^F(R^{\mathrm{lh}*})\\
        &=
        \left(1+\frac{\mu}{3n^{3}}\right) \left(1+\frac{3\mu^{3}}{1-\mu}\right)\cdot C_j^F(R^{\mathrm{lh}*})\le^* (1+2\mu) C_j^F(R^{\mathrm{lh}*}),
    \end{align*}
    The statement is then obtained by bounding $C^F_j(R^S)\le C^F_j(V)$ and summing over all $j\in J^{\mathrm{lh}}$.
\end{proof}

\subsection{Comparing a Line Schedule with its Slotted Counterpart}

In this sub-section, we prove \Cref{lem:error_continuous_instead_of_slotted}.
The first statement is straightforward to prove using \Cref{obs:job_scheduled_in_many_slots}:

\begin{proof}[Proof of Statement 1 of \Cref{lem:error_continuous_instead_of_slotted}]
    By the second statement of \Cref{obs:job_scheduled_in_many_slots}, there can be at most $n^2$ slots that contain intersections between two dual lines, or of a dual line with the time axis.
    By definition of \Cref{def:line_schedule} and the $LP$ slackness conditions, the resource assignment stays the same in all slots that do not contain such intersections.
    
    By the first statement of \Cref{obs:job_scheduled_in_many_slots}, each job $j\in J^{\mathrm{lh}}$ can only process $\mu^{4}/n^{3} v_j$ volume in such a slot.
    If follows that, $j$ can at most lose $\mu^{4}/n v_j$ volume compared to a schedule corresponding to the primal solution that corresponds to $\alpha$.
    Hence $\bar{v}\ge v_j-\mu^{4}/n v_j=(1-\mu^{4}/n)v_j$.
\end{proof}

The second statement is more difficult to prove.
We first use \Cref{obs:job_scheduled_in_many_slots} to establish that
there are at most $n^2$ slots where $R^{\mathrm{lh}}$ and $R^S$ differ.
Similiar to the proof of \Cref{lem:slotting_instead_of_CLP_error_helper},
we want to essentially ignore these slots, i.e., bound their cost in terms of the cost of the other slots.

Because of \Cref{lem:slotting_instead_of_CLP_error_helper}, this is easy for the first $f$ slots.
For the other slots,
we require another \lcnamecref{lem:cost_rate}.
First, let us define the \emph{cost rate} of a resource distribution $R(t)$ as $c(R(t))\coloneqq \sum_{j\in J}{R_j(t)/v_j}$.
\Cref{lem:cost_rate} states then that the cost rate must be non-increasing with time $t$.

\begin{lemma}
    \label{lem:cost_rate}
    Let $(R^{\mathrm{lh}},\cdot,\cdot,\cdot,\bar{v})$ be a line schedule.
    Then $c(R^{\mathrm{lh}}(\cdot))$ is monotonically non-increasing.
\end{lemma}

The proof is given below.
Now we are able to bound the cost of a slot in terms of the cost of earlier slots, since the former has a smaller cost rate.
Now we want to assign each slot where $R^{\mathrm{lh}}$ and $R^S$ differ that is not among the first $f$ slots a sufficient number of earlier slots that in total have a much higher cost.
In short, we have to find a mapping as guaranteed by the following \lcnamecref{lem:n2_slots_negligible}.

\begin{lemma}
    \label{lem:n2_slots_negligible}
    Consider a slot set $\bar{I}$ with $\abs{\bar{I}}\le n^2$ and $i> f$ for all $i\in I$.
    Then we can find a map $M: \bar{I}\rightarrow P(I\setminus \bar{I})$ with $\forall i'\in M(i): i-2n^{3}/\mu<i'<i$ such that $M(i)\cap M(i')=\varnothing$ for all $i\ne i'$ and $\abs{M(i)}= n/\mu$.
\end{lemma}

Using \Cref{lem:cost_rate, lem:n2_slots_negligible}, we are now able to prove the second statement of \Cref{lem:error_continuous_instead_of_slotted}:

\begin{proof}[Proof of Statement 2 of \Cref{lem:error_continuous_instead_of_slotted}]
    Let $c_i$ be the cost of $R^{\mathrm{lh}}$ in slot $i>0$, i.e.,
    \[
    c_i\coloneqq \int_{i\delta}^{(i+1)\delta} \sum_{j\in J^{\mathrm{lh}}}\frac{R_j(t)}{v_j} t\dt=\int_{i\delta}^{(i+1)\delta} c(R^{\mathrm{lh}}(t)) t\dt,
    \]
    such that $C^F(R^{\mathrm{lh}})=\sum_{i=0}^{\infty} c_i$.
    We subdivide the set of slots into three sets, namely
    \begin{enumerate}
        \item {$I_{\mathrm{first}}$: The first $f$ slots}
        \item {$\bar{I}$: The set of all slots $i\notin I_{\mathrm{first}}$ where $R^{\mathrm{lh}}$ and $R^S$ differ}
        \item {$I_{\mathrm{rem}}$: all remaining slots}
    \end{enumerate}

    We first show that the cost of the slots in $\bar{I}$ is negligible:
    For that, consider any slot $i\in\bar{I}$. Then, using \Cref{lem:cost_rate}, we can first bound $c_i$:
    \begin{align*}
        c_i&=\int_{i\delta}^{(i+1)\delta} c(R^{\mathrm{lh}}(t)) t\dt\le
        \int_{i\delta}^{(i+1)\delta} c(R^{\mathrm{lh}}(i \delta)) (i+1)\delta\dt=c(R^{\mathrm{lh}}(i \delta)) (i+1)\delta^2
    \end{align*}
    On the other hand, if we use the mapping $M$ from \Cref{lem:n2_slots_negligible}, we can show that $c_i$ has a much smaller cost:
    \begin{align*}
        \sum_{i'\in M(i)} c_{i'}
        &=\sum_{i'\in M(i)} \int_{i'\delta}^{(i'+1)\delta} c(R^{\mathrm{lh}}(t)) t\dt
        \ge
        \sum_{i'\in M(i)}\int_{i'\delta}^{(i'+1)\delta} c(R^{\mathrm{lh}}((i'+1)\delta)) i' \delta\dt\\
        &=
        \sum_{i'\in M(i)} c(R^{\mathrm{lh}}(i\delta)) i' \delta^2
        \ge 
        c(R^{\mathrm{lh}}(i\delta))\cdot \delta^2 \sum_{i'\in M(i)} (i-2n^{3}/\mu)=\frac{n}{\mu}\cdot (i-2n^{3}/\mu)
    \end{align*}
    As such, we can establish for their ratio:
    \begin{align*}
        \frac{c_i}{\sum_{i'\in M(i)} c_{i'}}
        \le \frac{i+1}{\frac{n}{\mu}\cdot (i-2n^{3}/\mu)}
        \le \frac{f+1}{\frac{n}{\mu}\cdot (f-2n^{3}/\mu)}
    \end{align*}
    When we insert $f$, we get
    \begin{align*}
        \frac{\frac{3n^{3}}{\mu}+1}{\frac{n}{\mu}\cdot \left(\frac{3n^{3}}{\mu}-2n^{3}/\mu\right)}
        &=
        \frac{3n^{3}+\mu}{n^{3}}\cdot \frac{\mu}{n}
        \le 4\mu
    \end{align*}

    Thus, we established that the cost of slots $I_{\mathrm{first}}$ is negligible (\Cref{lem:slotting_instead_of_CLP_error_helper}) and the cost of $\bar{I}$ is negligible:
	\begin{align*}
		\sum_{i\in I_{\mathrm{first}}} c_i\le \frac{\mu^{4}}{n^{3}}\cdot f\cdot \frac{1}{1-\mu}\cdot \left(
		\sum_{i\in \bar{I}} c_i
		+\sum_{i\in I_{\mathrm{rem}}} c_i
		\right)
        =
        \frac{3\mu^{3}}{1-\mu}\cdot \left(
		\sum_{i\in \bar{I}} c_i
		+\sum_{i\in I_{\mathrm{rem}}} c_i
		\right)
	\end{align*}
	\begin{align*}
		\sum_{i\in \bar{I}} c_i
		\le
		4\mu \left(\sum_{i\in I_{\mathrm{first}}} c_i
		+\sum_{i\in I_{\mathrm{rem}}} c_i\right)
	\end{align*}
    Summing both, we obtain
    \begin{align*}
        \sum_{i\in I_{\mathrm{first}}} c_i+\sum_{i\in \bar{I}} c_i\le \frac{3\mu^{3}}{1-\mu}\cdot \left(
		\sum_{i\in \bar{I}} c_i
		+\sum_{i\in I_{\mathrm{rem}}} c_i
		\right)+4\mu\left(\sum_{i\in I_{\mathrm{first}}} c_i
		+\sum_{i\in I_{\mathrm{rem}}} c_i\right).
    \end{align*}
    Rearranging this yields
    \begin{align*}
        \left(\frac{3\mu^3}{1-\mu}+4\mu\right) \sum_{i\in I_{\mathrm{rem}}} c_i
        &\ge
        \left(1-4\mu\right)\sum_{i\in I_{\mathrm{first}}} c_i+\left(1-\frac{3\mu^3}{1-\mu}\right)\sum_{i\in \bar{I}} c_i\\
        &\ge
        \left(1-4\mu-\frac{3\mu^3}{1-\mu}\right)\left(\sum_{i\in I_{\mathrm{first}}} c_i+\sum_{i\in \bar{I}} c_i\right)
    \end{align*}

    For the total cost, we get that

    \begin{align*}
        C^F(R^{\mathrm{lh}})&=
        \sum_{i=0}^{\infty} c_i
		=
		\sum_{i\in I_{\mathrm{first}}} c_i
		+\sum_{i\in \bar{I}} c_i
		+\sum_{i\in I_{\mathrm{rem}}} c_i\\
        &\le
        \left(1+\frac{\frac{3\mu^3}{1-\mu}+4\mu}{1-4\mu-\frac{3\mu^3}{1-\mu}}\right)\sum_{i\in I_{\mathrm{rem}}} c_i
        \le^* (1+5\mu) C^F(R^S),
    \end{align*}
    where the last inequality is due to $R^S$ and $R^{\mathrm{lh}}$ having the same cost in all
    slots from $I_{\mathrm{rem}}$.
\end{proof}

Lastly, we give the proofs of \Cref{lem:cost_rate, lem:n2_slots_negligible}.

\begin{proof}[Proof of \Cref{lem:cost_rate}]
    By \Cref{lem:dual_schedule_completion}, $(R^{\mathrm{lh}},\cdot,\cdot,\cdot,\bar{v})$ is a primal-dual pair and therefore fulfills the slackness conditions for $CLP(\bar{v})$/$DCP(\bar{v})$.
    By \Cref{def:line_schedule}, we have $R(t)=R(t')$ if $\succ_t$ and $\succ_{t'}$ induce the same total order and the same set of dual lines lies above the time axis.
    The cost rate may only change at some time point $t$ if the resource distribution changes at that point.
    By definition, this is only possible if
    \begin{enumerate}
        \item {$d_j(t)=0$ for some job $j$ or}
        \item {$d_j(t)=d_{j'}(t)$ for two jobs $j,j'$.}
    \end{enumerate}
    In the first case, $R(\cdot)$ does only change at $t$ when $j$ was scheduled before $t$. In this case, by \Cref{def:line_schedule}, only $R^{\mathrm{lh}}_j(t)$ drops to zero, decreasing $c(R^{\mathrm{lh}}(\cdot))$ at $t$.
    In the second case, two jobs $j,j'$ exchange some volume $r>0$, i.e.,
    $R^{\mathrm{lh}}_j(\cdot)$ reduces by $r$, and $R^{\mathrm{lh}}_{j'}$ increases by $r$.
    This can only happen when $v_j<v_{j'}$ since all dual lines are monotonically decreasing.
    Therefore, the cost rate will change by $r\cdot (-{1}/{v_j}+{1}/{v_{j'}})<0$ at time $t$, again decreasing the cost rate.
    If more than two jobs meet at $t$, then we can express the change in cost rate as a multiple such changes.
\end{proof}

\begin{proof}[Proof of \Cref{lem:n2_slots_negligible}]
    We build $M$ from $\bar{I}$ by considering each $i\in \bar{I}$
    in descending order.
    For each such $i\in\bar{I}$, we define
    $i$ as the largest $(n/\mu)$ slots that are smaller than $i$ and not from $\bar{I}$.
    Since $\bar{I}$ does not contain any of the first $f$ slots
    and since $n^2\cdot (n/\mu)\le f$,
    we have enough slots to assign.
    By this assignment, all desired properties of $M$ are fulfilled,
    except we still have to show that $\forall i'\in M(i): i-2n^3/\mu<i'$ holds for any $i\in \bar{I}$:
    Since $\abs{\bar{I}}\le n^2$ and we choose the largest $n/\mu$ slots for each set $M(i)$, we have that $\forall i'\in M(i): i'\ge i-n^2-n^2\cdot (n/\mu)\ge i-2n^{3}/\mu$.
\end{proof}
%\clearpage
%\input{999-graveyard.tex}

% \printbibliography
\bibliography{references}

\end{document}